\documentclass[submission,copyright,creativecommons]{eptcs}
\providecommand{\event}{HCVS 2020} 
\usepackage{underscore}           

\usepackage[utf8]{inputenc}

\newif\ifdraft
\draftfalse

\usepackage{amsmath}
\usepackage{amsfonts}
\usepackage{amssymb}
\usepackage{amsthm}
\usepackage{mathtools}
\mathtoolsset{centercolon}

\usepackage{enumitem}

\usepackage{tabularx}

\usepackage{tikz}
\usepackage{tikz-qtree}

\usepackage{caption}
\usepackage{subcaption}

\usepackage{algorithm}
\usepackage{algorithmicx}
\usepackage[noend]{algpseudocode}

\usepackage{listings}

\lstdefinestyle{mystyle}{
    basicstyle=\ttfamily\footnotesize,
    keywordstyle=\bfseries,
    breakatwhitespace=true,
    breaklines=true,
    keepspaces=true,
    showspaces=false,
    showstringspaces=false,
    morekeywords={data,let,in}
}

\usepackage[capitalise,english]{cleveref}

\theoremstyle{plain}
\newtheorem{theorem}{Theorem}[section]
\newtheorem*{theorem*}{Theorem}
\newtheorem{lemma}[theorem]{Lemma}

\theoremstyle{definition}
\newtheorem{definition}[theorem]{Definition}
\newtheorem{example}[theorem]{Example}

\newcommand{\fin}{\operatorname{fin}}
\newcommand{\infin}{\operatorname{infin}}
\newcommand{\true}{\operatorname{true}}
\newcommand{\false}{\operatorname{false}}
\newcommand{\depth}{\operatorname{depth}}

\newcommand{\data}[2]{\mathbf{\mathsf{data}}\; #1 = #2}

\newcommand{\con}[1]{\mathsf{#1}}

\newcommand{\sel}[2]{\mathbf{sel}_{\con{#1}}^{#2}}

\AtBeginDocument{%
 \abovedisplayskip=3pt plus 3pt minus 3pt
 \abovedisplayshortskip=0pt plus 3pt
 \belowdisplayskip=3pt plus 3pt minus 3pt
 \belowdisplayshortskip=2pt plus 3pt minus 2pt
}

\title{The Extended Theory of Trees and Algebraic (Co)datatypes}
\author{
Fabian Zaiser \qquad\qquad C.-H. Luke Ong
\institute{University of Oxford}
\email{\quad fabian.zaiser@cs.ox.ac.uk \quad\qquad luke.ong@cs.ox.ac.uk}
}
\def\titlerunning{The Extended Theory of Trees and Algebraic (Co)datatypes}
\def\authorrunning{F. Zaiser \& C.-H. L. Ong}
\begin{document}
\maketitle

\begin{abstract}
The first-order theory of finite and infinite trees has been studied since the eighties, especially by the logic programming community.
Following Djelloul, Dao and Frühwirth, we consider an extension of this theory with an additional predicate for finiteness of trees, which is useful for expressing properties about (not just datatypes but also) codatatypes.
Based on their work, we present a simplification procedure that determines whether any given (not necessarily closed) formula is satisfiable, returning a simplified formula which enables one to read off all possible models.
Our extension makes the algorithm usable for algebraic (co)datatypes, which was impossible in their original work due to restrictive assumptions.
We also provide a prototype implementation of our simplification procedure and evaluate it on instances from the SMT-LIB.
\end{abstract}

\section{Introduction}

Trees play a fundamental role in computer science:
syntactic terms can be regarded as finite trees, and operations like matching and unification, which are essential to functional and logic programming languages, can be viewed as solving certain first-order constraints in the structure of finite trees.
Furthermore, trees are a model for program schemes, such as higher-order recursion schemes \cite{ong15}, and more generally as computation trees.
The structures of finite and infinite trees are also central to the declarative semantics of logic (e.g.~\cite{llo87}) and functional languages (e.g.~\cite{tur87}).
Furthermore, they play a role in the verification of programs \cite{opp80} and in term rewriting systems \cite{jk86}.
The theory of finite and infinite trees was extensively studied by the logic programming community in the eighties.
An axiomatization and a decision procedure for these structures was given by Maher in \cite{maher88,maher88full} and another decision procedure independently by Comon and Lescanne in \cite{cl89}.

The structure of trees just consists of what one would normally think of as trees with labeled nodes, except that we allow them to be infinite.
Examples of finite and infinite trees are depicted in \cref{fig-trees}.
The labels for these trees are suggestively named after constructors for common algebraic datatypes because we want to specifically consider applications to the theory of (co)datatypes.
In functional programming, two common data structures are natural numbers and linked lists:
\begin{lstlisting}[]
data nat = zero | succ(pred: nat)
data list = nil | cons(head: nat, tail: list)
\end{lstlisting}
Inhabitants of these types are naturally viewed as trees: the term \lstinline|cons(succ(zero), cons(zero, nil))| is shown as a tree in \cref{fig-tree-finite}.
In some languages, such as Haskell, datatypes behave in fact more like codatatypes \cite{rb17}, i.e. they can be infinitely nested.
For example, the term \lstinline|let t = cons(zero, cons(succ(zero), t)) in t| corresponds to the infinite tree shown in \cref{fig-tree-rational}.

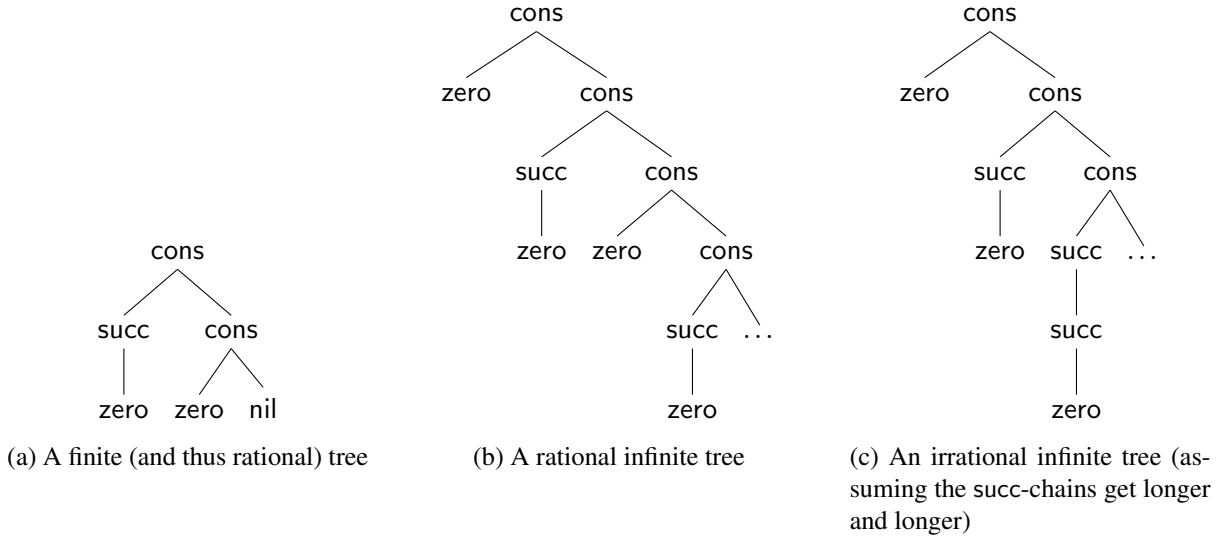
\begin{figure}
\centering
\begin{subfigure}[t]{0.3\textwidth}
\centering
\begin{tikzpicture}
\Tree [.$\con{cons}$ [.$\con{succ}$ $\con{zero}$ ] [.$\con{cons}$ $\con{zero}$ $\con{nil}$ ] ]
\end{tikzpicture}
\caption{A finite (and thus rational) tree}
\label{fig-tree-finite}
\end{subfigure}
\hfill
\begin{subfigure}[t]{0.3\textwidth}
\centering
\begin{tikzpicture}
\Tree [.$\con{cons}$ $\con{zero}$ [.$\con{cons}$ [.$\con{succ}$ $\con{zero}$ ] [.$\con{cons}$ $\con{zero}$ [.$\con{cons}$ [.$\con{succ}$ $\con{zero}$ ] $\dots$ ] ] ] ]
\end{tikzpicture}
\caption{A rational infinite tree}
\label{fig-tree-rational}
\end{subfigure}
\hfill
\begin{subfigure}[t]{0.3\textwidth}
\centering
\begin{tikzpicture}
\Tree [.$\con{cons}$ $\con{zero}$
    [.$\con{cons}$
        [.$\con{succ}$ $\con{zero}$ ]
        [.$\con{cons}$
            [.$\con{succ}$ [.$\con{succ}$ $\con{zero}$ ] ]
            $\dots$
        ]
    ]
]
\end{tikzpicture}
\caption{An irrational infinite tree (assuming the $\con{succ}$-chains get longer and longer)}
\label{fig-tree-infinite}
\end{subfigure}
\caption{Examples of trees}
\label{fig-trees}
\end{figure}

In this work, we consider the first-order theory of trees, extended with a predicate $\fin(t)$ for stating finiteness of $t$, and propose it as a tool for reasoning about algebraic datatypes and codatatypes.
Why not use the theory of (co)datatypes as implemented in many SMT solvers?
First, the theory of (co)datatypes is undecidable because selectors (\lstinline|head, tail|) can be applied to the wrong constructor (\lstinline|nil|) and the standard semantics from the SMT-LIB \cite{smt-lib} does not specify the result of such an operation (cf. \cref{sec-relationship}).
Secondly, the theory of trees not only allows us to treat datatypes and codatatypes in a uniform way (both are represented by trees, datatypes just require a $\fin$-predicate), but it is even more expressive: it can explicitly state non-finiteness ($\lnot \fin(t)$), as well as finiteness only for a proper subformula (cf. \cref{sec-relationship}).
Codatatypes are often used in mechanized proofs to represent infinite structures \cite{rb17} and we believe that this increased flexibility can be useful there as well.

This extended theory of trees (with a $\fin$-predicate) was first presented by Djelloul, Dao and Frühwirth in \cite{ddf08}, where they also give a complete axiomatization and a decision procedure.
However, one of their core assumptions is that there are infinitely many function symbols (i.e. constructors), which renders it unsuitable for algebraic (co)datatypes found in programming languages since those never have infinitely many constructors.
In this work, we lift this restriction and allow sorts with finitely many generators.
Note that we require at least two generators, however.
Sorts with one generator are not hard to support in principle but require a lot of special casing, so we do not discuss them in this work.

\vspace{-1em}
\paragraph{Contributions} Our first contribution is a formal description of the relationship between the theory of algebraic (co)datatypes and the extended theory of trees.
We also present a decision procedure for first-order formulae (including quantifiers) in the latter, based on \cite{ddf08}.
To the best of our knowledge, no decision procedure for this theory allowing finitely generated sorts was known before.
Just like the algorithm in \cite{ddf08}, it is not only a decision procedure that outputs ``satisfiable" or ``unsatisfiable" but instead, it simplifies the given input formula as much as possible, which makes it easy to read off all satisfying valuations of the free variables.

We propose this extended theory of trees as an interesting background theory for constrained Horn clauses.
Recently, Ong and Wagner \cite{ow19} proved that satisfiability of higher-order constrained Horn clauses (HoCHCs) is semi-decidable if the background theory is decidable.
Since the extended theory of trees is decidable, it is potentially suitable as a background theory for HoCHCs.
In addition, we hope that the existence of a simplification procedure instead of a mere decision procedure will also have useful applications to constrained Horn clauses.

\vspace{-1em}
\paragraph{Outline} The paper is organized as follows.
In \cref{sec-trees}, we explain the theory of trees.
The following \cref{sec-relationship} introduces (co)datatypes and explores their relationship with the theory of trees.
In \cref{sec-fin-gen-sorts}, we describe how to check for (finitely generated) sorts with only finitely many finite, respectively infinite, trees.
This step is necessary for the extension of the algorithm by Djelloul, Dao and Frühwirth \cite{ddf08} to finitely generated sorts.
The extended algorithm is presented in \cref{sec-solver}.
Throughout this paper, a sequence of mathematical objects $x_1,\dots,x_n$ is abbreviated as $\bar x$.
Unless otherwise stated, $u,v,w,x,y,z$ stand for variables, $s$ for sorts, $t$ for terms, $\phi,\psi$ for logical formulae, and $\con{f}, \con{g}$ for constructors and generators.

\section{Trees}
\label{sec-trees}

An \emph{ordered tree} is defined as a (potentially infinite) connected directed acyclic graph with a distinguished node $r$ (the \emph{root}) such that every vertex has exactly one incoming edge, except for $r$, which has none.
Additionally, for each node, its outgoing edges (and the corresponding nodes) are ordered.
Furthermore, each node is labeled with an element of a \emph{label set}, also called a \emph{function symbol} or a \emph{constructor}.
For a node $v$, its \emph{subtree} rooted at $v$ is the induced subgraph containing exactly the nodes reachable from $v$.
A tree is called \emph{finite} if it has finitely many nodes, and \emph{rational} if it has finitely many distinct subtrees.
Examples of trees can be found in \cref{fig-trees}.
These use labels suggestively named $\con{zero}, \con{succ}$ (representing natural numbers) and $\con{nil}, \con{cons}$ (representing lists).
This alludes to the connection with algebraic datatypes mentioned in the introduction and further explored in \cref{sec-relationship}.

The logical setting for the theory of trees is many-sorted first-order logic.
We have a set of \emph{sorts} $S$, a set of \emph{function symbols} $F$, a set of \emph{predicate symbols} $P$ and a countable set of \emph{variables} $V$.
Each function symbol $\con{f}$ has an \emph{arity} $\con{f}: s_1 \times \cdots \times s_n \to s$ with $s_1,\dots,s_n,s \in S$ and we say that $\con{f}$ is a \emph{generator} of $s$.
If $n = 0$, we write $\con{f}: s$ and say that $\con{f}$ is a \emph{constant}.
The set of generators of $s$ is called $F_s$ and we assume that each $F_s$ is nonempty.
We say that $s$ is \emph{finitely generated} if $F_s$ is finite, and \emph{singular} if $F_s$ is a singleton.
Similarly, each predicate symbol $p$ has an arity $p: s_1 \times \cdots \times s_n$ with $s_1,\dots,s_n \in S$, and each variable $v$ has a sort $v:s$.
For example, the trees in \cref{fig-trees} can be expressed in the first-order language with sorts $S = \{ \mathit{nat}, \mathit{list} \}$ and function symbols $\con{zero}: \mathit{nat}$, $\con{nil}: \mathit{list}$, $\con{succ}: \mathit{nat} \to \mathit{nat}$, and $\con{cons}: \mathit{list} \times \mathit{list} \to \mathit{list}$.
In the first-order language of trees, the only predicates are $P = \{ \fin_s: s \mid s \in S \}$, which state that a given tree of sort $s$ is finite.
We will drop the index $s$ if there is no ambiguity.

Given such a \emph{signature} $(S, F, P)$, a many-sorted \emph{structure} $\mathcal A$ consists of non-empty sets $s^\mathcal A$ for each $s \in S$, functions $\con{f}^\mathcal A: s_1^\mathcal A \times \cdots \times s_n^\mathcal A \to s^\mathcal A$ for each function symbol $\con{f}: s_1 \times \cdots \times s_n \to s$ and a predicate $\fin_s^\mathcal A \subseteq s^\mathcal A$ for each predicate symbol $\fin_s$.
A \emph{valuation} for $\mathcal A$ is a family of mappings $V_s \to s^\mathcal A$, indexed by $S$, where $V_s$ denotes the variables of sort $s$.
A \emph{model} of a formula $\phi$ in $\mathcal A$ is a valuation making $\phi$ true in $\mathcal A$.

The structure $\mathcal T$ of trees interprets a signature $(S, F, P)$ as follows.
Each sort $s$ is interpreted as the set $s^\mathcal T$ of \emph{trees of sort $s$}, meaning the trees where the root is labeled with a generator $\con{f}: s_1 \times \cdots \times s_n \to s$ and its children (in order) are roots of subtrees of sorts $s_1,\dots,s_n$, respectively.
Each function symbol $\con{f}: s_1 \times \cdots s_n \to s$ is interpreted as a function $\con{f}^\mathcal T: s_1^\mathcal T \times \cdots \times s_n^\mathcal T \to s^\mathcal T$ such that $\con{f}^\mathcal T(t_1,\dots,t_n)$ is the tree with a root labeled $\con{f}$ and subtrees $t_1,\dots,t_n$.
Each predicate $\fin_s$ is interpreted as the subset $\fin_s^\mathcal T \subseteq s^\mathcal T$ of finite trees of sort $s$.
Thus $\fin_s(t)$ holds in $\mathcal T$ if and only if the interpretation of $t$ is a finite tree.
For example, if $t_1$ is the term $\con{cons}(\con{succ}(\con{zero}), \con{cons}(\con{zero},\con{nil}))$, depicted in \cref{fig-tree-finite}, then $\fin(t_1)$ is true in $\mathcal T$.
On the other hand, if $t_2$ is the unique tree to make $t_2 = \con{cons}(\con{zero},\con{cons}(\con{succ}(\con{zero}),t_2))$ true (this tree is shown in \cref{fig-tree-rational}) then $\fin(t_2)$ is false because this tree is infinite.

We call the \emph{theory} of $\mathcal T$, i.e. the set of sentences that are true of $\mathcal T$, the \emph{extended theory of trees} (``extended" because of the additional predicate $\fin$).
This theory was first presented by Djelloul, Dao and Frühwirth \cite{ddf08}.
However, they require that each sort have at least one constant generator and infinitely many non-constant generators.
This assumption simplifies the treatment of the theory significantly but has the serious drawback that it makes their method unsuitable for algebraic (co)datatypes, which typically have only finitely many constructors (i.e. generators).
Therefore, we first take a look at the relation between the two theories.

\section{Relationship with (Co)Datatypes}
\label{sec-relationship}

The theory of \emph{algebraic (co)datatypes}, also called \emph{(co)inductive datatypes}, is similar to the theory of trees but there are a couple of important differences.
For one thing, the set of sorts is partitioned into $S = S_{dt} \cup S_{ct}$ where $S_{dt}$ is the set of datatypes and $S_{ct}$ is the set of codatatypes.
The function symbols are partitioned into the constructors $F_{ctr}$ and the selectors $F_{sel}$, and there are no predicate symbols.
Each (co)datatype $\delta$ is equipped with $m \geq 1$ constructors $F_{ctr}^\delta = \{ \con C_1, \dots, \con C_m \}$.
Each constructor $\con C_i \in F_{ctr}^\delta$ has an arity $\delta_1 \times \cdots \times \delta_{n_i} \to \delta$ and is associated with $n_i$ selectors $\sel{C_i}{j}: \delta \to \delta_j$.
Note that for a datatype (resp.\ codatatype) declaration, all constructor arguments must be datatypes (resp.\ codatatypes); no mixing is allowed.
Such an assumption is common, for example in \cite{rb17}.
Another requirement is that datatypes be well-founded, i.e. one must be able to exhibit a ground term for each datatype.
This excludes examples like a datatype $\mathit{infinite}$ with a single constructor $\con{next}: \mathit{infinite} \to \mathit{infinite}$.
However, the latter is allowed as a \emph{codatatype}.

\begin{example}
Consider Booleans and lists.
Their declaration in many programming languages looks roughly like this:
\begin{align*}
\data{&\mathit{bool}}{\con{True} \mid \con{False}} \\
\data{&\mathit{list}} {\con{Nil} \mid \con{Cons}(head: \mathit{bool}, tail: list)}
\end{align*}
where the selectors are called $head := \sel{Cons}{1}$ and $tail := \sel{Cons}{2}$.
The former extracts the first element of a given list, if it is nonempty, and the latter returns the rest of the list.
Hence we have $F_{ctr}^{\mathit{bool}} = \{ \con{True}, \con{False} \}$, $F_{sel}^{\mathit{bool}} = \emptyset$, $F_{ctr}^{\mathit{list}} = \{ \con{Nil}, \con{Cons} \}$, and $F_{sel}^{\mathit{list}} = \{ head, tail \}$.
\end{example}

\paragraph{Semantics}
Both datatypes and codatatypes are interpreted as \emph{constructor trees}, i.e. sorted trees labeled only with constructors (not selectors).
A structure $\mathcal D$ of (co)datatypes interprets a codatatype $\gamma$ as the set $\gamma^\mathcal D$ of constructor trees of sort $\gamma$ and a datatype $\delta$ as the set $\delta^\mathcal D$ of finite constructor trees of sort $\delta$.
Each constructor $\con C: s_1 \times \dots \times s_n \to s$ is interpreted as the function $\con C^\mathcal D: s_1^\mathcal D \times \dots \times s_n^\mathcal D \to s^\mathcal D$ constructing a new tree out of the given ones, with root $\con C$.
Each selector $\sel C i: s \to s_i$ is interpreted as a function $(\sel C i)^\mathcal D: s^\mathcal D \to s_i^\mathcal D$, which must satisfy $(\sel C i)^\mathcal D(\con C^\mathcal D(t_1, \dots, t_n)) = t_i$ but is not specified on inputs built with the wrong constructor.
This semantics is very common and what the SMT-LIB standard specifies~\cite{smt-lib}.
Note that other semantics are possible, however, such as returning a fixed default value if a selector is applied to the wrong constructor \cite{bst07}.
We call the latter the \emph{semantics with default values}.
The theory of (co)datatypes is the set of sentences that are true in any structure of (co)datatypes satisfying the above.

\begin{theorem}
\label{thm-codatatypes-undecidable}
The first-order theory of (co)datatypes is undecidable.
\end{theorem}
\begin{proof}[Proof idea]
The proof is based on the undecidability of formulae with quantifiers in the theory of uninterpreted functions (EUF).
Uninterpreted functions $f: s_1 \times s_2 \to s$ can be emulated using the following construct:
\[ \data{dummy}{\con c(a_1: s_1, a_2: s_2) \mid \con d(h: s)} \]
Then $h(\con c(x,y))$ acts like an uninterpreted function $f(x,y)$ because the selector $h$ is applied to the wrong constructor $\con c$.
For the full proof, refer to the appendix.
\end{proof}

In the semantics with default values, selectors present much less of a problem: they can simply be eliminated.

\begin{theorem}
\label{thm-eliminate-default-values}
In the theory of (co)datatypes with default values, a given formula can be effectively transformed into an equivalent one without selectors.
\end{theorem}
\begin{proof}[Proof idea]
The idea is to introduce additional variables such that selectors only occur in equations of the form $x = \sel C i(t)$ where $t$ does not contain selectors and then to rewrite such an equation as
\[ (\exists \bar v \ldotp t = C(\bar v) \land x = v_i) \lor ((\lnot\exists \bar v \ldotp t = \con C(\bar v)) \land x = T_\con C^i) \]
where $T_\con C^i$ is the default value for this selector.
For the full proof, refer to the appendix.
\end{proof}

But even in the standard semantics, quantifier-free formulae in the theory of (co)datatypes are decidable \cite{rb17}.
In fact, we can also eliminate selectors from such formulae.

\begin{theorem}
\label{thm-eliminate-standard}
In the theory of (co)datatypes with standard semantics, a quantifier-free formula can be effectively transformed into an equisatisfiable one without selectors (but including quantifiers).
\end{theorem}
\begin{proof}[Proof idea]
The first step is to introduce additional variables such that selectors only occur in equations of the form $x = \sel C i(t)$ where $t$ does not contain selectors.
For each such equation, we add the conjunct $\forall \bar z. t = \con C(\bar z) \to z_i = x$, which ensures that the selector correctly extracts the argument when applied to the right constructor.
Furthermore, for each pair of such equations $x = \sel C i(t)$, $x' = \sel C i(t')$, we add the conjunct $t = t' \to x = x'$, which ensures that selectors behave like functions, i.e. return the same result when applied to the same arguments.
For the full proof, refer to the appendix.
\end{proof}

Finally, we show that the theory of trees is enough for selector-free formulae.

\begin{theorem}
\label{thm-codatatypes-to-trees}
A selector-free formula in the theory of (co)datatypes can be effectively transformed into an equisatisfiable formula in the extended theory of trees.
\end{theorem}
\begin{proof}[Proof idea]
Since in both theories, terms are interpreted as trees, we just have to ensure that datatypes are interpreted as finite trees.
Hence, for a datatype $d$, existential quantification $\exists x:d \ldotp \phi$ is replaced by $\exists x:d \ldotp \fin(x) \land \phi$ and universal quantification $\forall x:d. \phi$ is replaced by $\forall x:d. \fin(x) \to \phi$.
Finally, to ensure equisatisfiability, free variables $x:d$ require adding the conjunct $\fin(x)$ to the whole formula.
For the full proof, refer to the appendix.
\end{proof}

The last result raises the question of how the expressiveness of the extended theory of trees compares to selector-free (co)datatypes.
The former is, in fact, more expressive because it allows specifying non-finiteness of individuals, such as $\lnot \fin(x)$.
This is impossible in the theory of (co)datatypes since datatypes have only finite values and codatatypes can have finite and infinite values.
Additionally, it facilitates specifying finiteness only in parts of the formula, such as in $(\fin(t) \to \phi) \lor(\lnot\fin(t) \to \psi)$, where $t$ is finite in $\phi$ but infinite in $\psi$.
This shows that the extended theory of trees is more powerful than the (selector-free) theory of (co)datatypes.

\section{Analyzing Finitely Generated Sorts}
\label{sec-fin-gen-sorts}

Having shown how formulae involving (co)datatypes can often be reduced to formulae involving trees, we want to find a decision procedure for the latter based on the work by Djelloul, Dao and Frühwirth \cite{ddf08}.
Their algorithm, however, makes the assumption that each sort contains infinitely many non-constant generators and one constant generator.
As a consequence, each sort contains infinitely many finite and infinitely many infinite trees.
This simplifies solving logical formulae: the predicate $\fin(x)$ can always be made true or false for an appropriate valuation of $x$.
However, their assumption is obviously not satisfied for sorts arising from (co)datatypes.

Therefore, we consider the setting with finitely generated sorts, where the situation is more complicated.
For instance, if $x$ is of, say, a Boolean sort with only constant generators then the predicate $\fin(x)$ is always true.
Due to these complications, we need to analyze the set of sorts and check for sorts with only finitely many finite or infinite trees.

In the following, we allow sorts with finitely many generators but assume that any sort has at least two generators.
As mentioned before, this restriction is not hard to lift in principle but saves us a lot of technical details and space in this paper.
Note that sorts with a single non-recursive generator can just be unfolded in the place that they are used.

For a sort $s$, denote by $s_{\fin}$, respectively $s_{\infin}$, the set of finite, respectively infinite, trees of sort $s$.
Denote by $S_{0F}$, $S_{FF}$, $S_{0I}$, $S_{1I}$, $S_{FI} \subseteq S$ the sets of sorts with no finite trees, finitely many finite trees, no infinite trees, exactly one infinite tree, and finitely many infinite trees, respectively.
In the following, we present algorithms for computing these sets.

\begin{algorithm}
\caption{Algorithm computing the sets of sorts containing no finite ($S_{0F} \subseteq S$) or no infinite trees ($S_{0I} \subseteq S$), when given a signature $(S,F,P)$ as input.}
\label{alg-zero-sorts}
\begin{algorithmic}
\State $S_{0I} \gets \emptyset$
\State $S_{0F} \gets S$
\Repeat
\State $S_{0I} \gets S_{0I} \cup \{ s \in S \mid \forall (\con g: s_1 \times \cdots \times s_n \to s) \in F_s: \forall i \in \{ 1,\dots, n \}: s_i \in S_{0I} \} $
\State $S_{0F} \gets S_{0F} \setminus \{ s \in S \mid \exists (\con g: s_1 \times \cdots \times s_n \to s) \in F_s: \forall i \in \{ 1,\dots, n \}: s_i \notin S_{0F} \}$
\Until{no changes in the last iteration}
\end{algorithmic}
\end{algorithm}

\begin{theorem}
\label{thm-zero-sorts}
Given a signature $(S,F,P)$, \cref{alg-zero-sorts} correctly computes the sets $S_{0F}$ and $S_{0I}$.
\end{theorem}
\begin{proof}[Proof idea]
A sort contains no infinite trees if every generator only takes arguments of sorts containing no infinite trees.
A sort contains no finite trees unless some generator takes only arguments of sorts with finite trees.
The sets $S_{0F}$ and $S_{0I}$ can thus be computed as fixed points, the former a least fixed point, the latter a greatest fixed point.
For details, refer to the full proof in the appendix.
\end{proof}

\begin{figure}
\begin{align*}
\con{false}&: \mathit{bool} & \con{zero}&: \mathit{nat} & \con{nil}&: \mathit{list} \\
\con{true}&: \mathit{bool} & \con{succ}&: \mathit{nat} \to \mathit{nat} & \con{cons}&: \mathit{nat} \times \mathit{list} \to \mathit{list} \\[10pt]
\con{tree1}&: \mathit{inftree} \to \mathit{inftree} &\con{c1}&: \mathit{bool} \to d &\con{g1}&: \mathit{bool} \times \mathit{bool} \to t \\
\con{tree2}&: \mathit{inftree} \times \mathit{inftree} \to \mathit{inftree} &\con{c2}&: \mathit{nat} \times \mathit{inftree} \to d &\con{g2}&: \mathit{bool} \times \mathit{nat} \to t
\end{align*}
\caption{\label{fig-signature}Generators for the sorts $S = \{ \mathit{bool}, \mathit{nat}, \mathit{list}, \mathit{inftree}, d, t \}$.}
\end{figure}

\begin{example}
Consider the sorts and generators in \cref{fig-signature}.
How would \cref{alg-zero-sorts} act on this input?
In the first iteration, it would add $\mathit{bool}$ to $S_{0I}$ because each generator is a constant.
At the same time, $\mathit{bool}$, $\mathit{list}$, and $\mathit{nat}$ are removed from $S_{0F}$ because each one has a constant generator.
In the next iteration, $S_{0I}$ stays unchanged but $d$ is removed from $S_{0F}$ because it has the generator $\con{c1}$ whose parameter sort $\mathit{bool}$ is not in $S_{0F}$ anymore.
For a similar reason, $t$ is removed from $S_{0F}$.
After this point, no more changes happen and we obtain $S_{0I} = \{ \mathit{bool} \}$ and $S_{0F} = \{ \mathit{inftree} \}$.
\end{example}

Next, we consider the sets $S_{FF}, S_{1I}, S_{FI}$.
Note that the sort $\mathit{nat}$ from the above example has exactly one infinite tree, namely $\con{succ}(\con{succ}(\dots))$.
For such sorts $s \in S_{1I}$, we introduce variables $u_s$ for their unique infinite tree.
For instance, $u_{\mathit{nat}} = \con{succ}(u_{\mathit{nat}})$ describes the unique infinite tree of $\mathit{nat}$.
The sort $t$ has two infinite trees $\con{g2}(\con{false}, u_{\mathit{nat}})$ and $\con{g2}(\con{true}, u_{\mathit{nat}})$.
Hence to describe all infinite trees of sorts $s \in S_{FI}$, we need the variables $u_s$ for $s \in S_{1I}$ and their equations, like $u_{\mathit{nat}} = \con{succ}(u_{\mathit{nat}})$.
\Cref{alg-fin-ind} computes all this.

\begin{algorithm}[t]
\caption{Algorithm computing the sets of sorts containing only finitely many finite ($S_{FF}$), respectively infinite ($S_{FI}$), trees; and their finite ($s_{\fin}$), respectively infinite ($s_{\infin}$), inhabitants.}
\label{alg-fin-ind}
\begin{algorithmic}
\State Compute $S_{0I}$ and $S_{0F}$ as in \cref{alg-zero-sorts}
\State $S_{FF} \gets S_{0F}$
\State $s_{\fin} \gets \emptyset$ for each $s \in S$
\State $S_{1I} \gets S \setminus S_{0I}$
\State $U_s \gets \emptyset$ for each $s \in S$
\Repeat
\For{$s \in S$}
\State $F_s^{\infin} \gets \{ (\con g: s_1 \times \cdots \times s_n \to s) \in F_s \mid \exists i \in \{ 1, \dots, n \}: s_i \in S_{0F} \}$
\If{$|F_s \setminus F_s^{\infin}| < \infty$ and $\forall (\con g: s_1 \times \cdots \times s_n \to s) \in F_s \setminus F_s^{\infin}: \forall i \in \{ 1, \dots, n \}: s_i \in S_{FF}$}
\State $S_{FF} \gets S_{FF} \cup \{ s \}$
\State $s_{\fin} \gets \{ \con g(r_1,\dots,r_n) \mid (\con g: s_1 \times \cdots \times s_n \to s) \in F_s \setminus F_s^{\infin}, r_i \in (s_i)_{\fin} \}$
\EndIf
\If{$\exists (\con g: s_1 \to s) \in F_s: s_1 \in S_{1I} \land \big(\forall (\con{g'}: s_1' \times \cdots \times s_n' \to s) \in F_s \setminus \{ \con g \}: \forall i: s_i' \in S_{0I}\big)$}
\State $U_s \gets \{ u_s = \con g(u_{s_1}) \} \cup U_{s_1}$
\Else
\State $S_{1I} \gets S_{1I} \setminus \{ s \}$
\EndIf
\EndFor
\Until{no changes in the last iteration}

\State $S_{FI} \gets S_{0I} \cup S_{1I}$
\State $s_{\infin} \gets \emptyset$ for each $s \in S_{0I}$
\State $s_{\infin} \gets \{ u_s \}$ for each $s \in S_{1I}$
\Repeat
\For{$s \in S$}
\State $F_s^{\infin} \gets \{ (\con g: s_1 \times \cdots \times s_n \to s) \in F_s \mid \exists i \in \{ 1,\dots,n \}: s_i \notin S_{0I} \}$
\If{$|F_s^{\infin}| < \infty$ and $\forall (\con g: s_1 \times \cdots \times s_n \to s) \in F_s^{\infin}: \forall i \in \{ 1, \dots, n \}:$\\
\hspace{\algorithmicindent}\hspace{\algorithmicindent}\hspace{\algorithmicindent}\hspace{\algorithmicindent}$s_i \in S_{0I} \lor \big(s_i \in S_{FI} \land (\forall j \in \{ 1, \dots, n \} \setminus \{ i \}: s_j \in S_{FF} \cap S_{FI})\big)$}
\State $S_{FI} \gets S_{FI} \cup \{ s \}$
\State $s_{\infin} \gets \{ \con g(r_1, \dots, r_n) \mid (\con g: s_1 \times \cdots \times s_n \to s) \in F_s^{\infin}; r_j \in (s_j)_{\fin} \cup (s_j)_{\infin} \text{ for } j = 1, \dots, n$\\
\hspace{\algorithmicindent}\hspace{\algorithmicindent}\hspace{\algorithmicindent}\hspace{\algorithmicindent}\hspace{\algorithmicindent}\hspace{\algorithmicindent}\hspace{\algorithmicindent}\hspace{\algorithmicindent}\hspace{\algorithmicindent} such that $ \exists i \in \{1, \dots, n \}: r_i \notin (s_i)_{\fin} \}$
\EndIf
\EndFor
\Until{no changes in the last iteration}
\end{algorithmic}
\end{algorithm}

\begin{theorem}
\label{thm-fin-ind}
Given a signature $(S,F,P)$, \cref{alg-fin-ind} correctly computes the sets $S_{FF}$, $S_{1I}$, and $S_{FI}$.
Furthermore it computes the set $s_{\fin}$ (the terms for the finite trees of sort $s$ for $s \in S_{FF}$), and the set $s_{\infin}$ (the terms for the infinite trees of sort $s$ for $s \in S_{FI}$).
The latter makes use of the variables $u_s$ (for $s \in S_{1I}$), standing for the unique infinite tree of $s$.
The equations that uniquely determine these $u_s$ are output in $U_s$.
\end{theorem}
\begin{proof}[Proof idea]
Similarly to the previous algorithm, these sets are computed as fixed points.
A sort $s$ has only finitely many finite trees if there is a finite number of generators that only take sorts with finitely many finite trees as arguments ($F_s \setminus F_s^{\infin}$), and the remaining generators ($F_s^{\infin}$) take at least one argument of a sort that contains no finite trees (because such a generator cannot create finite trees).
Along the way, the algorithm builds up the set $s_{\fin}$ from the generators of the former category.

Constructing the set $S_{FI}$ works similarly, except for the fact that we start the fixed point iteration with $S_{0I} \cup S_{1I}$ instead of the empty set.
The reason is that for every sort with finitely many infinite trees, it can be shown (\cref{lem-shape-infinite-trees} in the appendix) that the infinite parts of each such tree are built from the unique infinite trees of the sorts $S_{1I}$.
These sorts with a unique infinite tree are also constructed by fixed point iteration.
They can only have a single generator $\con g$ that constructs infinite trees and it can only take one argument because otherwise we would have at least two infinite trees since each sort is assumed to have at least two generators.

A sort $s$ only has finitely many infinite trees if the set of generators constructing infinite trees ($F_s^{\infin}$) is finite, and when picking an arbitrary argument $i$ of it, this argument allows no infinite trees; or it allows finitely many infinite trees and all the other arguments allow only finitely many trees.
This explains the fixed point iteration for $S_{FI}$ and $s_{\infin}$.
For details, refer to the full proof in the appendix.
\end{proof}

\begin{example}
Consider again the signature from \cref{fig-signature}.
How does \cref{alg-fin-ind} act on it?
At the start of the first loop, we have $S_{FF} = \{ \mathit{inftree} \}$ and $S_{1I} = S$.
In the first iteration, $\mathit{bool}$ is added to $S_{FF}$ because all its generators are constants, and $\mathit{bool}_{\fin} = \{ \con{false}, \con{true} \}$.
Additionally, $\mathit{nat}$ stays in $S_{1I}$ because its generator $\con{succ}$ satisfies $\mathit{nat} \in S_{1I}$ and the other generator is constant.
Therefore $U_{\mathit{nat}} = \{ u_{\mathit{nat}} = \con{succ}(u_{\mathit{nat}}) \}$.
All the other sorts are removed from $S_{1I}$, either because they don't have a unary generator ($\mathit{bool}$, $\mathit{list}$) or there is another generator that allows infinite trees, destroying uniqueness ($\con{tree2}$ for $\mathit{inftree}$, $\con{c2}$ for $d$, and $\con{g2}$ for $t$).
In the second iteration, $d$ is added to $S_{FF}$ because $\con{c1}: \mathit{bool} \to d$ only constructs finitely many finite trees since $\mathit{bool} \in S_{FF}$ and its other generator $\con{c2}$ constructs only infinite trees.
Therefore, $d_{\fin}$ is set to $\{ \con{c1}(\con{true}), \con{c1}(\con{false}) \}$.
After this point, no more changes happen.

At the start of the second loop, we have $S_{FI} = \{ \mathit{bool}, \mathit{nat} \}$ and $\mathit{nat}_{\infin} = \{ u_{\mathit{nat}} \}$.
In the loop iteration, $t$ is added to $S_{FI}$ because we have $F_t^{\infin} = \{ \con{g2} \}$, which is finite, and its only generator $\con{g2}$ has the property that its first parameter is $\mathit{bool} \in S_{0I}$ and its second parameter is $\mathit{nat} \in S_{FI}$ with the additional property that all remaining parameters, i.e. $\mathit{bool}$, are in $S_{FF} \cap S_{FI}$.
Therefore $t_{\infin}$ is set to $\{ \con{g2}(\con{false}, u_{\mathit{nat}}), \con{g2}(\con{true}, u_{\mathit{nat}}) \}$.
After this point, no more changes happen.
The algorithm has computed $S_{FF} = \{ inftree, bool, d \}$ and $S_{FI} = \{ nat, t \}$.
\end{example}

\section{Simplification Procedure for the Extended Theory of Trees}
\label{sec-solver}

Having explained how to analyze finitely generated sorts, we can now describe how the simplification procedure from \cite{ddf08} is extended to finitely generated sorts.
Before going into detail, we provide a brief outline of this algorithm.
The procedure works on special formulae, called \emph{normal formulae}.
Any formula can be transformed into an equivalent normal formula, so this is not a restriction.
Roughly speaking, the output of our algorithm is a disjunction of \emph{fully simplified} formulae that is equivalent to the original formula.
A fully simplified formula makes it easy to read off all its models.
The simplification algorithm works similarly to \cite{ddf08}, except for the fact that finitely generated sorts sometimes require case splits (also called instantiations) for certain variables (called \emph{instantiable}).
These case splits can be on the finitely many generators of a sort, or on the finitely many (finite or infinite) inhabitants of a sort if it is in $S_{FF}$ or $S_{FI}$.
In this section, we focus on these instantiable variables and case splits because that is the novel part of our extension of \cite{ddf08}.
The full algorithm is described in the appendix (\cref{alg-solve-normal}).
Before we can start with the concept of normal formulae, we first need to define \emph{basic formulae}.

\begin{definition}
\label{def-basic}
A \emph{basic formula} is  of the form $(\bigwedge_i v_i = t_i) \land (\bigwedge_j \fin(u_j))$ where $\bar u,\bar v$ are variables and each $t_i$ is a variable or a term of the form $\con{f}(\bar z)$ for a function symbol $\con{f}$ and variables $\bar z$.
Such a formula will be abbreviated as $\overline{v = t} \land \overline{\fin(u)}$.
Given a total order on its variables, it is called \emph{solved} if (1) the variables $\bar u, \bar v$ are distinct and for each equation $x = y$, we have $x > y$, and (2) if $\fin(v)$ occurs then the sort of $v$ contains both finite and infinite trees.
A variable $x_n$ is \emph{reachable} from a variable $x_0$ if the basic formula contains
$x_0 = t_0 \land x_1 = t_1 \land \cdots \land x_{n-1} = t_{n-1}$
where each $t_i$ contains $x_{i+1}$.
It is \emph{properly reachable} if $n > 0$.
The subformulae $u = t$ and $\fin(u)$ are considered reachable if $u$ is.
\end{definition}

The variable ordering is important when we consider basic subformulae of larger formulae.
Then this ordering ensures that in solved basic formulae, variables bound deeper inside the whole formula occur on the left-hand side of equations, which is important for the correctness proof.
Intuitively, reachability means the following: if $y$ is reachable from $x$ then $y$ is a subtree of $x$.
Djelloul, Dao and Frühwirth describe an algorithm to solve a basic formula (rules 1--10 in \cite[Section 4.6]{ddf08}).
In our extended setting, two things have to be changed: if $\fin(u)$ occurs in the basic formula where $u:s$ with $s \in S_{0I}$, or $s \in S_{0F}$, then $\fin(u)$ is always satisfied and can be removed, or is never satisfied and the basic formula is unsolvable, respectively.
This is summarized by the following theorem.

\begin{theorem}
\label{thm-solve-basic}
There is an algorithm $\Call{solveBasic}{(v_0, \cdots, v_n), \alpha}$ (\cref{alg-solve-basic} in the appendix) that correctly solves basic formulae $\alpha$ containing the variables $v_0, \cdots, v_n$, i.e.\ it turns $\alpha$ into an equivalent solved formula (with respect to the variable ordering $v_0 < \cdots < v_n$) or returns $\false$ if none exists.
\end{theorem}

Basic formulae are insufficient for the general case but they are an important building block for the concept of \emph{normal formulae}, which can express any first-order formula.

\begin{definition}
A \emph{normal formula} $\phi$ of depth $d \geq 1$ takes the form $\lnot(\exists \bar x \ldotp \alpha \land \bigwedge_{i = 1}^n \phi_i)$ where $\alpha$ is a basic formula, and each $\phi_i$ is a normal formula of depth $d_i$ with $d = 1 + \max(0,d_1,\dots,d_n)$.
\end{definition}

The simplest normal formula is $\lnot \true$.
As normal formulae allow expressing negation, conjunction, existential quantification, and nesting, the next theorem is straightforward to prove \cite[Property 4.3.3]{ddf08}.

\begin{theorem}
\label{thm-normalize}
There is an algorithm $\Call{normalize}{\phi}$ which turns any first-order formula $\phi$ into a normal one that is equivalent in the theory of trees.
\end{theorem}

\begin{example}
Consider the formula $\forall x: \mathit{nat}. \lnot\fin(x) \to x = \con{succ}(x)$.
It can be rewritten as $\lnot(\exists x: \mathit{nat} \ldotp \lnot\fin(x) \land \lnot(x = \con{succ}(x)))$, which is a normal formula of depth 2.
\end{example}

Now we come to the main difference from the original algorithm in \cite{ddf08}: our more general setting necessitates case splits (or instantiations) for certain variables.
For instance, consider the normal formula $\phi_1 \equiv \lnot(\exists x: \mathit{list} \ldotp \lnot(x = \con{nil}) \land \lnot(\exists y,z \ldotp x = \con{cons}(y,z)))$.
If $\mathit{list}$ had infinitely many generators, it would always be possible to find a value for $x$ that is neither $\con{nil}$ nor $\con{cons}$.
However, since $\mathit{list}$ only has those two generators, no such $x$ exists and the formula is true.
Here our extended algorithm will do a case split on both constructors of $\mathit{list}$ (described later in more detail) and realize that neither works.

As another example, consider $\phi_2 \equiv \lnot(\exists x: t \ldotp \lnot\fin(x) \land \lnot(x = y) \land \lnot(x = z))$.
If $t$ had infinitely many infinite trees, then this would be true because we could always choose a valuation for $x$ that is different from the free variables $y$ and $z$.
Since $t$ has only two infinite trees, our extended algorithm does a case split on all two infinite trees of $t$, instantiating $x$ with $\con{g2}(\con{true}, u_{\mathit{nat}})$ and $\con{g2}(\con{false}, u_{\mathit{nat}})$ where $u_{\mathit{nat}}$ is the unique tree of sort $\mathit{nat}$, namely $\con{succ}(\con{succ}(\dots))$.
A similar case occurs with a constraint $\fin(x)$ where $x$ only has finitely many finite trees or, in general, if $x$ has only finitely many trees.
This leads us to the definition of an instantiable variable, i.e. a variable that requires a case split.

Note that every normal formula can be transformed into an equivalent conjunction of normal formulae of depth at most 2 by repeatedly applying rule 16 (depth reduction) from \cite[Section 4.6]{ddf08}.
Therefore we can limit our attention to such formulae in the following.

\begin{definition}[instantiable variable]
\label{def-instantiable}
Let $\lnot\exists \bar x \ldotp \alpha \land \bigwedge_i \lnot(\exists \bar y_i \ldotp \beta_i)$ be a normal formula of depth at most 2 such that $\alpha$ and each $\beta_i$ are solved basic formulae.
Let $\beta_i^*$ be $\beta_i$ with all conjuncts also occurring in $\alpha$ removed.
Then a variable $v: s$ that is free in the formula or occurs in $\bar x$ is called \emph{instantiable} if one of the following conditions is satisfied:
\begin{enumerate}[noitemsep, nolistsep]
\item $s$ has finitely many generators $F_s$ and some $\beta_i^*$ contains $v = \con{f}(\bar w)$ and $v$ is not properly reachable from $v$ in $\beta_i$, or
\item $s \in S_{FF} \cap S_{FI}$, some $\beta_i^*$ contains $v$, and $\alpha$ contains no equation $v = t$ for any term $t$, or
\item $s \in S_{FF}$, $\alpha$ contains $\fin(v)$, and some $\beta_i^*$ contains $v$, or
\item $s \in S_{FI}$ and some $\beta_i^*$ contains only $\fin()$-constraints (no equations), among them $\fin(v)$.
\end{enumerate}
\end{definition}

\begin{algorithm}
\caption{Algorithm for finding instantiable variables and their instantiations.}
\label{alg-find-instantiations}
\begin{algorithmic}
\Function{findInstantiation}{$\bar v, \lnot(\exists \bar x \ldotp \alpha \land \bigwedge_i \lnot(\exists \bar y_i \ldotp \beta_i))$}
\State $\beta_i^* \gets \beta_i$ without the conjuncts occurring in $\alpha$
\For{$u: s \in \bar v\bar x$}
    \If{$s$ has finitely many generators, and $u = \con{f}(\bar z)$ occurs in some $\beta_i^*$, \\
    \hspace{\algorithmicindent}\hspace{\algorithmicindent}\hspace{\algorithmicindent}\hspace{\algorithmicindent}and $u$ is not properly reachable from $u$ in $\beta_i$}
    \State \Return $\{ \exists \bar z \ldotp u = \con g(\bar z) \mid \con g \in F_s \}$
    \EndIf
    \If{$s \in S_{FF} \cap S_{FI}$ and $u$ occurs in some $\beta_i^*$ and $\alpha$ contains no $u = t$}
    \State \Return $\{ \exists u_{s_1},\dots,u_{s_n} \ldotp u = t \land \bigwedge_{s \in S_{1I}} U_s \mid t \in s_{\fin} \cup s_{\infin} \}$ where $S_{1I} = \{ s_1, \dots, s_n \}$
    \EndIf
    \If{$s \in S_{FF}$ and $u$ occurs in some $\beta_i^*$, and $\fin(u)$ in $\alpha$}
    \State \Return $\{ u = t \mid t \in s_{\fin} \}$
    \EndIf
    \If{$s \in S_{FI}$ and $\fin(u)$ occurs in some $\beta_i^*$ that contains only $\fin()$-constraints}
    \State \Return $\{ \fin(u) \} \cup \{ \exists u_{s_1},\dots,u_{s_n} \ldotp u = t \land \bigwedge_{s \in S_{1I}} U_s \mid t \in s_{\infin} \}$  where $S_{1I} = \{ s_1, \dots, s_n \}$
    \EndIf
\EndFor
\Return $none$
\EndFunction
\end{algorithmic}
\end{algorithm}

\Cref{alg-find-instantiations} looks for an instantiable variable (if any) in a normal formula $\lnot(\exists \bar x \ldotp \alpha \land \bigwedge_i \phi_i)$ with free variables $\bar v$ by checking exactly the four conditions from above.
If it finds an instantiable variable $u$, it returns a set $I$ of formulae, called \emph{instantiations}.
Note that while we write ``$u = t$" in the return values for simplicity, we actually mean an equivalent formula $\exists \bar z \ldotp \gamma$ where $\gamma$ is a basic formula.
For instance, by $u = \con{c1}(\con{true})$, we mean $\exists z \ldotp u = \con{c1}(z) \land z = \con{true}$ for a fresh variable $z$.
We can use these instantiations to get rid of instantiable variables, as the following theorem explains.

\begin{theorem}
\label{thm-instantiations}
Let $\phi \equiv \lnot(\exists \bar x \ldotp \alpha \land \bigwedge_i \phi_i)$ be a normal formula of depth at most 2 with free variables $\bar v$.
Let $I$ be the result of $\Call{findInstantiation}{\bar v, \phi}$ from \cref{alg-find-instantiations}.
If $I$ is ``none", then there is no instantiable variable.
Otherwise, let $u$ be the first instantiable variable found in \Call{findInstantiation}{}.
Then $\phi$ is equivalent to the following conjunction of normal formulae, in which the variable $u$ is no longer instantiable:
\[ \bigwedge_{(\exists \bar z \ldotp \psi) \in I} \lnot(\exists \bar x\bar z \ldotp \alpha \land \psi \land \bigwedge_i \phi_i). \]
\end{theorem}

\begin{example}
In the formula $\phi_1 \equiv \lnot(\exists x: \mathit{list} \ldotp \lnot(x = \con{nil}) \land \lnot(\exists y,z \ldotp x = \con{cons}(y,z)))$ from above, $x$ is instantiable because of condition 1.
(Note that the reachability check in this condition is required to avoid infinite loops for recursive equations like $x = \con{cons}(y,x)$.)
Here \Call{findInstantiation}{} returns $I = \{ x = \con{nil}; \exists y,z \ldotp x = \con{cons}(y,z) \}$.
By the above theorem, $\phi_1$ is equivalent to
\begin{align*}
&\lnot(\exists x: \mathit{list} \ldotp x = \con{nil} \land \lnot(x = \con{nil}) \land \lnot(\exists y,z \ldotp x = \con{cons}(y,z)) \\
\land & \lnot(\exists x: \mathit{list},y,z \ldotp x = \con{cons}(y,z) \land \lnot(x = \con{nil}) \land \lnot(\exists y,z \ldotp x = \con{cons}(y,z)).
\end{align*}
Both existential subformulae obviously contain a contradiction, so the whole formula simplifies to $\lnot(\false) \land \lnot(\false)$ and thus $\true$.
\end{example}

\begin{example}
In the other formula $\phi_2 \equiv \lnot(\exists x: t \ldotp \lnot\fin(x) \land \lnot(x = y) \land \lnot(x = z))$ from above, $t$ is instantiable because of condition 4.
\Cref{alg-find-instantiations} returns the instantiations
\[ I = \{ \fin(x); \exists u_{\mathit{nat}} \ldotp x = \con{g2}(\con{true}, u_{\mathit{nat}}) \land u_{\mathit{nat}} = \con{succ}(u_{\mathit{nat}}); \exists u_{\mathit{nat}} \ldotp x = \con{g2}(\con{true}, u_{\mathit{nat}}) \land u_{\mathit{nat}} = \con{succ}(u_{\mathit{nat}}) \}, \]
which means $x$ is either a finite tree or one of the two infinite trees $\con{g2}(\con{false}, u_{\mathit{nat}})$, $\con{g2}(\con{true}, u_{\mathit{nat}})$ where $u_{\mathit{nat}}$ is the unique tree with $u_{\mathit{nat}} = \con{succ}(u_{\mathit{nat}})$.
By the above theorem, $\phi_2$ is equivalent to
\begin{align*}
&\lnot(\exists x: t \ldotp \fin(x) \land \lnot\fin(x) \land \lnot(x = y) \land \lnot(x = z)) \\
\land &\lnot(\exists x, u_{\mathit{nat}}: t \ldotp x = \con{g2}(\con{false}, u_{\mathit{nat}}) \land u_{\mathit{nat}} = \con{succ}(u_{\mathit{nat}}) \land \lnot\fin(x) \land \lnot(x = y) \land \lnot(x = z)) \\
\land &\lnot(\exists x, u_{\mathit{nat}}: t \ldotp x = \con{g2}(\con{true}, u_{\mathit{nat}}) \land u_{\mathit{nat}} = \con{succ}(u_{\mathit{nat}}) \land \lnot\fin(x) \land \lnot(x = y) \land \lnot(x = z)).
\end{align*}
The other parts of the simplification procedure (unchanged from \cite{ddf08}) simplify this to
\begin{align*}
&\true \\
\land &\lnot(\exists u_{\mathit{nat}}: t \ldotp u_{\mathit{nat}} = \con{succ}(u_{\mathit{nat}}) \land \lnot(y = \con{g2}(\con{false}, u_{\mathit{nat}})) \land \lnot(z = \con{g2}(\con{false}, u_{\mathit{nat}}))) \\
\land &\lnot(\exists u_{\mathit{nat}}: t \ldotp u_{\mathit{nat}} = \con{succ}(u_{\mathit{nat}}) \land \lnot(y = \con{g2}(\con{true}, u_{\mathit{nat}})) \land \lnot(z = \con{g2}(\con{true}, u_{\mathit{nat}}))),
\end{align*}
where the variable $x$ is removed because it is unreachable from the free variables.
The resulting formula essentially expresses that $y$ or $z$ has to be equal to $\con{g2}(\con{false}, u_{\mathit{nat}})$; and that $y$ or $z$ has to be equal to $\con{g2}(\con{true}, u_{\mathit{nat}})$.
In other words, they can only take on those two infinite values and have to be different.
Note that the algorithm has not completed at this point yet because $y$ and $z$ are now instantiable by condition 1.
We skip the following (less interesting) instantiations for space reasons.
\end{example}

At this point, we can introduce the notions of \emph{solved} and \emph{fully simplified} formulae, which make up the output of our extended simplification procedure.

\begin{definition}
\label{def-solved}
A normal formula $\phi \equiv \lnot(\exists \bar x \ldotp \alpha \land \bigwedge_i \lnot(\exists \bar y_i \ldotp \beta_i))$, of depth at most 2, is called \emph{solved} if it satisfies the following properties.
\begin{enumerate}[noitemsep, nolistsep]
\item Each $\beta_i$ and $\alpha$ are solved basic formulae with respect to a variable ordering where $u < v$ if the binding of $v$ is more deeply nested than $u$, i.e. $u$ is free where $v$ is bound.
\item The equations of $\alpha$ are included in every $\beta_i$.
\item Each $\beta_i$ contains at least one conjunct that does not occur in $\alpha$.
\item There are no instantiable variables.
\item All the variables $\bar x$ and $\bar y_i$ are reachable from the free variables of $\exists \bar x \ldotp \alpha$ and $\exists \bar y_i \ldotp \beta_i$, respectively.
\end{enumerate}
A formula $\psi$ is called \emph{fully simplified} if $\lnot \psi$ is a solved normal formula.
(This is an extension of the definition of ``explicit solved form" in \cite[Definition 4.4.6]{ddf08}.)
\end{definition}

\begin{example}
The point of fully simplified formulae is that they are easy to interpret, i.e. it is easy to read off all possible models from them.
For instance, consider the fully simplified formula
\[ \exists v \ldotp x = \con{succ}(v) \land v = y \land \fin(y) \land \lnot(\exists w \ldotp y = \con{succ}(w) \land \fin(w) \land \fin(z)) \]
Any model has to satisfy $x = \con{succ}(y)$ and $y$ has to be finite.
To falsify the other part $\exists w \ldotp y = \con{succ}(w) \land \fin(w) \land \fin(z)$, there are two options for the free variables $y$ and $z$: (1) instantiate $y$ with any finite tree with a root other than $\con{succ}$ and $z$ with any tree, or (2) instantiate $y$ with any finite tree and $z$ with any infinite tree.
These are the only two classes of models for the above fully simplified formula.
In general, the following holds about fully simplified formulae.
\end{example}

\begin{theorem}
\label{thm-solved-form}
Let $\phi$ be a fully simplified formula.
If $\phi$ has no free variables then $\phi \equiv \true$.
Otherwise both $\phi$ and $\lnot \phi$ are satisfiable in the theory of trees.
\end{theorem}

Given a formula $\phi$, the main algorithm returns an equivalent disjunction of fully simplified formulae.
Since each disjunct allows an easy description of its models, we can describe all possible models of $\phi$.

\begin{theorem}
\label{thm-solver-correct}
There is an algorithm $\Call{solve}{\phi}$ (\cref{alg-solve-normal} in the appendix) that, given a formula $\phi$, returns $\true$, $\false$, or a disjunction of fully simplified formulae that is equivalent to $\phi$ in the extended theory of trees.
In particular, if $\phi$ is closed, it returns $\true$ or $\false$.
\end{theorem}
\begin{proof}[Proof idea]
The simplification procedure from \cite{ddf08} does not have to be changed a lot.
We use the function \Call{solveBasic}{} from \cref{thm-solve-basic} to solve basic formulae.
Afterwards, the rules 12--14 and 16 from \cite[Section 4.6]{ddf08} ensure that the result is a disjunction of formulae satisfying conditions (1--3) of \cref{def-solved}.
At this point, we make use of \Call{findInstantiation}{} and \cref{thm-instantiations}, to ensure that condition (4) is satisfied.
After each such instantiation, the previous rules have to be applied again because conditions (1--3) may have been invalidated.
These steps will not, however, invalidate condition (4).
It is nontrivial to prove that these instantiations terminate (\cref{lem-termination} in the appendix).
Avoiding infinite loops of instantiations is the reason for the complicated condition 1 in \cref{def-instantiable}.
Finally, Rule 15 of the original algorithm \cite[Section 4.6]{ddf08} ensures that condition (5) is satisfied as well.
It is again nontrivial to show that this part is still correct in our more general setting (\cref{lem-rule15-correct} in the appendix).
\end{proof}

\vspace{-1.5em}
\paragraph{Time complexity} Regarding the performance of our extended algorithm, note that the original algorithm has non-elementary time complexity \cite{ddf08}.
In fact, Vorobyov proved that deciding first-order formulae in the (ordinary) theory of trees already has non-elementary time complexity \cite{vorobyov96}, so we cannot hope for an efficient algorithm for the extended theory of trees in the worst case.

\vspace{-1em}
\paragraph{Implementation} In order to evaluate the performance in practice, we created a prototype implementation \cite{implementation}.
Due to a lack of benchmarks involving formulae of trees, we took the \texttt{tests} set of the \texttt{QF\_DT} (quantifier-free datatypes) suite of the SMT-LIB \cite{smt-lib} and transformed each instance into a formula in the extended theory of trees, using the results from \cref{sec-relationship}.
Then we ran our extended simplification procedure on the transformed instances.
Over 90\% of them completed in less than 1 second and about 5\% timed out after 10 seconds (more data in \cref{fig-benchmark-results} in the appendix).
While state-of-the-art SMT solvers decide each \texttt{QF\_DT} instance in a few milliseconds, our transformed instances are considerably harder because they contain quantifiers and can be significantly larger than the original ones.
Furthermore, our prototype implementation obviously cannot compete with heavily optimized SMT solvers and leaves a lot of room for improvements: for instance, the normalization of formulae can be optimized, and heuristics for choosing which variable to instantiate (instead of picking the first one) could make a big difference.
It nevertheless demonstrates that our algorithm has a reasonable performance on many practical instances.
Our implementation can be found at \url{https://github.com/fzaiser/tree-theory-solver/} or \cite{implementation}.
Additionally, there is a web interface at \url{http://mjolnir.cs.ox.ac.uk/trees-codata/}.

\section{Conclusion}

We believe that the extended theory of trees is an interesting theory to study because of its decidability and connections with algebraic (co)datatypes.
We have explained the complications arising from finitely generated sorts, which are necessary to apply it to (co)datatypes.
The fact that we not only provide a decision procedure but a simplification procedure should make it easier to conduct further research on the theory of trees, such as investigating Craig interpolation.

\bibliographystyle{eptcs}
\bibliography{literature}

\clearpage
\appendix

\section{Omitted proofs from Section \ref{sec-relationship}}

\begin{theorem*}[\cref{thm-codatatypes-undecidable}, repeated]
The first-order theory of (co)datatypes is undecidable.
\end{theorem*}
\begin{proof}[Proof of \cref{thm-codatatypes-undecidable}]
The proof works by reduction from Post's correspondence problem.
An instance of the problem is given by a finite set of pairs of bit strings $\{(x_1,y_1),\dots,(x_n,y_n)\}$, i.e. each $x_i,y_i \in \{0,1\}^*$.
A solution to such an instance is a nonempty finite sequence of indices $i_1,\dots,i_k$ with each $i_j \in \{1,\dots,n\}$ such that
\[ x_{i_1}\dots x_{i_k} = y_{i_1}\dots y_{i_k}. \]
The decision problem is to decide whether such a solution exists,
and is famously undecidable \cite{post46}.
Let $\{(x_1,y_1),\dots,(x_n,y_n)\}$ be an instance of the problem.
In the following, we construct a signature and formula $\phi$, which is satisfiable in the theory of datatypes if and only if the instance has a solution.

For the encoding, consider the following datatype declarations.
\begin{align*}
\data{\mathit{bool}&}{\con{true} \mid \con{false}} \\
\data{\mathit{bitstring}&}{\con e \mid \con 0(tail_0: \mathit{bitstring}) \mid \con 1(tail_1: \mathit{bitstring})} \\
\data{dummy&}{\con f(b_1: \mathit{bitstring}, b_2: \mathit{bitstring}) \mid \con{g}(h: \mathit{bool})}
\end{align*}
We have the usual definition of Booleans (true or false),
bit strings are empty ($\con e$) or start with a 0 or with a 1,
and the dummy datatype is just used to emulate uninterpreted functions: $h(\con f(s,s'))$ acts like an unspecified function on the bit strings $s,s'$ because $h$ is applied to the wrong constructor and the standard semantics leaves this case unspecified.
In this proof, $\phi$ is constructed in such a way that $h(\con f(s,s')) = true$ holds if $s,s'$ is a valid pair of bit strings that can be built out of the $(x_i,y_i)$.

First, for a bit string $x$, we write $p_x(t)$ for the bit string that prepends $x$ to $t$.
For example, if $x = 101$, $p_x(t)$ stands for $\con 1(\con 0(\con 1(t)))$.
Consider the following formulae.
\begin{align*}
\phi_1 &\equiv \bigwedge_{i=1}^n h(\con f(p_{x_i}(\con e), p_{y_i}(\con e))) = \con{true} \\
\phi_2 &\equiv \forall u,v. h(\con f(u,v)) = \con{true} \to \bigwedge_{i=1}^n h(\con f(p_{x_i}(u), p_{y_i}(v))) = \con{true} \\
\phi_3 &\equiv \exists u \ldotp h(\con f(u,u)) = \con{true}
\end{align*}
First, $\phi_1$ and $\phi_2$ specify what pairs of bit strings can be constructed out of $x_i$ and $y_i$,
and finally $\phi_3$ specifies that a solution to the instance exists.
Hence we claim that $\phi :\equiv \phi_1 \land \phi_2 \to \phi_3$ is valid if and only if the given instance of Post's correspondence problem is solvable.

First, suppose $\phi$ is valid.
Consider the model $\mathcal M$ where $h$ is interpreted as follows:
\[ h^\mathcal M(z) = \begin{cases}
y &\text{if } z = \con g^\mathcal M(y) \\
\con{true}^\mathcal M &\text{if } z = \mathcal M(\con f(p_x(\con e),p_y(\con e))) \\
&\text{ where } \exists k\ge 1, i_1,\dots,i_k: x = x_1\dots x_k \land y = y_1\dots y_k \\
\con{false}^\mathcal M &\text{otherwise.}
\end{cases} \]
In this model, $\phi_1$ and $\phi_2$ are satisfied by construction, hence $\phi_3$ is satisfied as well.
This means there is an element $u \in \mathit{bitstring}^\mathcal M$ such that $h^\mathcal M(\con f^\mathcal M(u,u)) = \con{true}^\mathcal M$.
By the choice of $\mathcal M$, this means that $u = \mathcal M(p_x(\con e))$ such that there are $i_1,\dots,i_k$ with $x = x_{i_1}\dots x_{i_k}$ and $y = y_{i_1}\dots y_{i_k}$, which means the given instance is solvable.

Conversely, suppose the instance of Post's correspondence problem has a solution $i_1,\dots,i_k$.
Let $\mathcal M$ be a structure where $\phi_1 \land \phi_2$ is true.
Let $s_j = x_{i_1}\dots x_{i_j}$ and $t_j = y_{i_1}\dots y_{i_j}$.
Then $\phi_1$ ensures that $h(\con f(p_{s_1}(\con e),\allowbreak p_{t_1}(\con e))) = \con{true}$ in $\mathcal M$
and $\phi_2$ ensures that the induction step works and hence that $h(\con f(p_{s_j}(\con e),\allowbreak p_{t_j}(\con e))) = \con{true}$ holds for all $j$, in particular for $j = k$.
Set $u = p_{s_k}(\con e)$.
Since $s_k = t_k$ is a solution to the instance, we have $h(\con f(u, u)) = \con{true}$ in $\mathcal M$.
Hence $\phi_3$ also holds in $\mathcal M$.
Therefore $\phi$ is valid.

Since a formula is valid if and only if its negation is unsatisfiable, this proof also shows that the satisfiability problem of first-order formulae in the theory of datatypes is undecidable.

Note that in the proof, we never used the fact that a $bit string$ is finite, so the same proof works when replacing the above datatype declarations by codatatype declarations.
\end{proof}

\begin{theorem*}[\cref{thm-eliminate-default-values}, repeated]
In the theory of (co)datatypes with default values, a given formula can be effectively transformed into an equivalent one without selectors.
\end{theorem*}
\begin{proof}[Proof of \cref{thm-eliminate-default-values}]
The transformation works in two steps.
First, it ensures that selectors don't occur nested but instead only occur in equations of the form $x = \sel C i(t)$ for a variable $x$.
Afterwards, it replaces such equations by an equivalent formula that doesn't contain selectors.

\emph{Step 1.}
The first step isolates selectors to simple equations.
For every equation $t = t'$ where one side contains a selector, do the following.
If $t = t'$ is of the form $x = \sel C i(s)$ for a variable $x$ and a selector-free term $s$ then there is nothing to do.
Otherwise, without loss of generality, assume that $t$ contains a (nested) selector, so $t$ can be written as $r[\sel C i(s)/x]$ for a term $r$ where $x$ is a fresh variable.
Then rewrite $r[\sel C i(s)/x] = t'$ to $\exists x \ldotp r = t' \land x = \sel C i(s)$, which is always equivalent in first-order logic.
This process is repeated until no more changes can be made.

\emph{Step 2.}
Now the formula only contains selectors as part of equations of the form $x = \sel C i(s)$ where $s$ is selector-free.
Each such equation can be rewritten as
\[ (\exists \bar v \ldotp s = \con C(\bar v) \land x = v_i)
\lor ((\lnot \exists \bar v \ldotp s = \con C(\bar v)) \land x = T_\con C^i) \]
where $\bar v$ are fresh variables and $T_\con C^i$ is the default value for this selector.
This is equivalent because
\begin{align*}
x = \sel C i(s) &\leftrightarrow ((\exists \bar v \ldotp s = \con C(\bar v)) \land x = \sel C i(s)) \lor (\lnot (\exists \bar v \ldotp s = \con C(\bar v)) \land x = \sel C i(s)) \\
&\leftrightarrow (\exists \bar v \ldotp s = \con C(\bar v) \land x = \sel C i(s)) \lor (\lnot (\exists \bar v \ldotp s = \con C(\bar v)) \land x = \sel C i(s)) \\
&\leftrightarrow (\exists \bar v \ldotp s = \con C(\bar v) \land x = v_i) \lor (\lnot (\exists \bar v \ldotp s = \con C(\bar v)) \land x = T_\con C^i).
\end{align*}
The first equivalence is simply a case split on $\exists \bar v \ldotp s = \con C(\bar v)$, the second one just moves a quantifier outward and the last one uses the definition of selectors.
\end{proof}

\begin{theorem*}[\cref{thm-eliminate-standard}, repeated]
In the theory of (co)datatypes with standard semantics, a quantifier-free formula can be effectively transformed into an equisatisfiable one without selectors (but including quantifiers).
\end{theorem*}
\begin{proof}[Proof of \cref{thm-eliminate-standard}]
The transformation works in two steps.
The first step is similar to the previous proof and replaces every selector term by an existentially quantified fresh variable.
The second step is different from before obviously, but is similar to Ackermann's reduction for uninterpreted functions \cite{ack54} (see \cite[Section 11.2.1]{ks16} for an exposition).

\emph{Step 1.}
Let $\phi$ be the formula to transform.
Let $S$ be a set of equations, initially empty.
As long as $\phi$ contains a selector term, pick one that doesn't contain any nested selector terms.
Let that term be $t \equiv \sel C i(s)$.
Replace $t$ by a fresh variable $v$ in $\phi$ and add the equation $v = \sel C i(s)$ to $S$.
This is repeated until $\phi$ contains no more selector terms.

\emph{Step 2.}
Let $\phi'$ be the formula after the first step and $S = \{ v_j = \sel{C_j}{i_j}(s_j) \mid j \in \{1,\dots,n\} \}$ the set of equations of selectors.
For each $j = 1,\dots,n$ define the following two formulae
\begin{align*}
\psi_j &\equiv \forall \bar z. s_j = \con{C_j}(\bar z) \to z_{i_j} = v_j \\
\chi_j &\equiv \bigwedge_{j'=1}^n s_j = s_{j'} \to v_j = v_{j'}
\end{align*}
where $\bar z$ are fresh variables for each $j$.
The idea is that the $\psi_j$'s specify that $\sel{C}{i}(\con C(\bar z)) = z_i$ and the $\chi_j$'s (which is also part of Ackermann's reduction) ensure functional consistency, i.e. that if the same selector is applied to equal terms then the results should be equal.
The latter property is important if the selector $\sel C i$ is applied to a term that is not of the form $\con C(\dots)$ because then the result is unspecified but it still has to be consistent.
Finally define $\psi$ to be the following formula
\[ \psi \equiv \bigwedge_{j=1}^n \left(\psi_j \land \chi_j \right). \]

Let $\phi$ be the original formula and $\phi'$ be the result of step 1.
We claim that $\psi \land \phi'$, which does not contain selectors anymore, is equisatisfiable to $\phi$.

First, suppose that $\phi$ is satisfiable.
Let $\mathcal M$ be a model of $\phi$.
We modify it to a model $\mathcal M'$ by just changing the valuation of some variables, leaving the interpretation of selector functions untouched:
\[ \mathcal M'(v) = \begin{cases}
\mathcal M'(\sel {C_j} {i_j}(s_j)) &\text{if } v \equiv v_j \\
\mathcal M(v) &\text{otherwise.}
\end{cases} \]
Note that the first case is well-defined since there are no cyclic dependencies between the variables $\bar v$ by the design of Step 1.
Then $\mathcal M'$ satisfies each $\psi_j$ since the interpretation of $\sel C i$ in $\mathcal M$ (and hence $\mathcal M'$) has to satisfy the selector axioms.
It also satisfies each $\chi_j$ because $\mathcal M'(\sel C i)$ is a function that has to give the same result if applied to the same object.
Hence if $\mathcal M'(s_j) = \mathcal M'(s_{j'})$ then $\mathcal M'(v_j) = \mathcal M'(v_{j'})$.
This shows that $\mathcal M'$ satisfies $\psi$.

Furthermore, it is not hard to see that $\mathcal M'$ satisfies $\phi'$ as well.
First of all, $\mathcal M'$ satisfies $\phi$ because the variables $v_j$ don't occur in $\phi$ since they were assumed to be fresh and $\mathcal M'$ agrees with $\mathcal M$ on everything else.
Also note that by construction of $\mathcal M'$, every equation $v_j = \sel{C_j}{i_j}(s_j)$ in $S$ is true of $\mathcal M'$.
Since each $\sel{C_j}{i_j}(s_j)$ in $\phi$ is replaced by $v_j$ in $\phi'$, the interpretations of $\phi$ and $\phi'$ are the same in $\mathcal M'$.
Hence $\phi'$ is true of $\mathcal M'$.
Altogether, $\psi \land \phi'$ is satisfiable.

Next, suppose that $\psi \land \phi'$ is satisfiable and let $\mathcal M$ be a model of it.
We define a new model $\mathcal M'$ modifying the interpretation of the selectors such that it satisfies:
\[
\mathcal M'(\sel C i)(x) = \begin{cases}
y_i &\text{if } x = \con C^\mathcal M(y_1,\dots,y_n) \\
\mathcal M(v_j) &\text{if } x = \mathcal M(s_j) \text{ for some } j \in \{1,\dots,n\} \\
\mathcal M(\sel C i)(x) &\text{otherwise.}
\end{cases}
\]
It is not immediately clear that this is well-defined.
First, we check what happens if there is more than one $j$ such that $x = \mathcal M(s_j)$.
Suppose $x = \mathcal M(s_j) = \mathcal M(s_{j'})$.
Since $\mathcal M$ satisfies $\psi$ and in particular $\chi_j$, we have $\mathcal M(v_j) = \mathcal M(v_{j'})$.
Hence it does not matter which $j$ is chosen.

Next, we check that the first and second case are compatible.
Suppose $x = \con C^\mathcal M(y_1,\dots,y_n) = \mathcal M(s_j)$.
Since $\mathcal M$ satisfies $\psi$ and in particular $\psi_j$, we have $y_i = \mathcal M(v_j)$.
So the first and second case are not in conflict.
Altogether, this shows that $\mathcal M'$ is well-defined.

Furthermore, it is not hard to see that $\mathcal M'$ is a model of $\phi$.
First of all, $\mathcal M'$ satisfies $\phi'$ because $\mathcal M$ does, $\phi'$ does not contain any selectors, and $\mathcal M'$ agrees with $\mathcal M$ on everything else.
Also note that by construction of $\mathcal M'$, every equation $v_j = \sel{C_j}{i_j}(s_j)$ in $S$ is true of $\mathcal M'$.
Since each $\sel{C_j}{i_j}(s_j)$ in $\phi$ is replaced by $v_j$ in $\phi'$, the interpretations of $\phi$ and $\phi'$ are the same in $\mathcal M'$.
Hence $\phi$ is true of $\mathcal M'$, and thus $\phi$ is satisfiable, which finishes the proof.
\end{proof}

\begin{theorem*}[\cref{thm-codatatypes-to-trees}, repeated]
A selector-free formula in the theory of (co)datatypes can be effectively transformed into an equisatisfiable formula in the extended theory of trees.
\end{theorem*}
\begin{proof}[Proof of \cref{thm-codatatypes-to-trees}]
Without loss of generality, we can assume that the given formula $\phi$ contains only existential quantifiers, no universal ones, because $\forall x. \psi$ can be replaced by $\lnot\exists x \ldotp \lnot \psi$.
Let $\phi'$ be the result of replacing each $\exists v_1,\dots,v_n \ldotp \psi$ occurring in $\phi$ with
\[ \exists v_1:s_1,\dots,v_n:s_n \ldotp \left(\bigwedge_{i\in\{1,\dots,n\}, s_i \text{ is a dataype}} \fin(v_i)\right) \land \psi. \]
Finally let $u_1: s_1,\dots,u_n: s_n$ be the free variables of $\phi'$.
Set
\[ \phi'' \equiv \left(\bigwedge_{i\in\{1,\dots,n\}, s_i \text{ is a dataype}} \fin(u_i)\right) \land \phi'. \]
Then it is clear that each model of $\phi$ can be reduced to a model of $\phi''$ by forgetting the interpretation of the selector functions because the interpretation of datatypes are finite trees, thus satisfying all the additional $\fin()$-constraints.
Conversely, each model of $\phi''$ in the theory of trees can be extended to a model of $\phi$ by picking arbitrary selector functions.
The $\fin()$-constraints ensure that the interpretation of each datatype variable is a finite tree, so it is, in fact, a model of $\phi$.
\end{proof}

\section{Omitted proofs from Section \ref{sec-fin-gen-sorts}}

Note that we make the standard assumption that the first-order language of trees is well-founded, in the sense that there is no infinite sequence $s_0,s_1,\dots$ of sorts such that for all $i \in \mathbb{N}$, there is a function symbol $f_i$ of arity $f_i: \cdots \times s_{i+1} \times \cdots \to s_i$.

\begin{theorem*}[\cref{thm-zero-sorts}, repeated]
Given a signature $(S,F,P)$, \cref{alg-zero-sorts} correctly computes the sets $S_{0F}$ and $S_{0I}$.
\end{theorem*}
\begin{proof}[Proof of \cref{thm-zero-sorts}]
The validity of the fixed point computations is implied by the following two lemmas.
\end{proof}

\begin{lemma}
\label{theorem-only-fin}
The set $S_{0I}$ of sorts without infinite trees is the least fixed point of the following function $f: P(S) \to P(S)$ where $P(\cdot)$ denotes the power set:
\[ f(X) := X \cup \{ s \in S \mid \forall (\con g: s_1 \times \cdots \times s_n \to s) \in F_s: \forall i \in \{1,\dots,n\}: s_i \in X \} \]
\end{lemma}
\begin{proof}
Note that $f$ is monotonic, so the least fixed point is guaranteed to exist by the Knaster-Tarski theorem.
It is clear that $S_{0I}$ is a fixed point because for each generator $\con g$ each parameter sort must only allow finite trees, i.e. be in $S_{0I}$, yielding $f(S_{0I}) = S_{0I}$.
We claim that any sort $s$ containing only finite trees of depth $\leq d$ is included in the set $f^{d+1}(\emptyset)$.
This is proved inductively.
For $d = 0$, this means $s$ contains only constant symbols.
Then the condition $\forall \con g: s_1 \times \cdots \times s_n \to s \in F_s: \forall i \in \{1,\dots,n\}: s_i \in \emptyset$ is vacuously true and $s \in f(\emptyset)$.
For $d \geq 1$, it means that each argument of each generator has depth at most $d - 1$.
Hence by the induction hypothesis, all generator arguments have sorts in $f^d(\emptyset)$.
By the definition of $f$, this means that then $s \in f(f^d(\emptyset))$, proving the claim.
Hence $S_{0I} = \cup_{d \in \omega} f^{d+1}(\emptyset)$, in other words, $S_{0I}$ is the least fixed point of $f$.
\end{proof}

\begin{lemma}
The set $S_{0F}$ of sorts without finite trees is the greatest fixed point of the following function $f: P(S) \to P(S)$:
\[ f(X) := X \setminus \{ s \in S \mid \exists (\con g: s_1 \times \cdots \times s_n \to s) \in F_s: \forall i \in \{1,\dots,n\}: s_i \notin X \} \]
\end{lemma}
\begin{proof}
Note again that $f$ is monotonic, so the greatest fixed point is guaranteed to exist by the Knaster-Tarski theorem.
It is clear that $S_{0F}$ is a fixed point because a sort with a generator $\con g$ where all parameter sorts of $\con g$ allow finite trees cannot contain only infinite trees.
Hence all such sorts must be excluded, which is what $f$ does.
Therefore $S_{0F} = f(S_{0F})$.
Similarly to the proof of \cref{theorem-only-fin}, it is easy to see by induction that any sort $s$ containing a finite tree of depth $d$ is excluded from the set $f^{d+1}(S)$.
Hence $S_{0F} = \cap_{d \in \omega} f^{d+1}(S)$, in other words, $S_{0F}$ is the greatest fixed point of $f$.
\end{proof}

\begin{theorem*}[\cref{thm-fin-ind}, repeated]
Given a signature $(S,F,P)$, \cref{alg-fin-ind} correctly computes the sets $S_{FF}$, $S_{1I}$, and $S_{FI}$.
Furthermore it computes the set $s_{\fin}$ (the terms for the finite trees of sort $s$ for $s \in S_{FF}$), and the set $s_{\infin}$ (the terms for the infinite trees of sort $s$ for $s \in S_{FI}$).
The latter makes use of the variables $u_s$ (for $s \in S_{1I}$), standing for the unique infinite tree of $s$.
The equations that uniquely determine these $u_s$ are output in $U_s$.
\end{theorem*}
\begin{proof}[Proof of \cref{thm-fin-ind}]
The correctness of the fixed point computation of $S_{FF}$ is implied by \cref{lem-SFF}.
An analogous argument verifies the computation of $s_{\fin}$, for $s \in S_{FF}$.
\Cref{lem-S1I} shows the correctness of the fixed point computation of $S_{1I}$ and the $U_s$.
For the correctness proof of the fixed point computation of $S_{FI}$, we need \cref{lem-shape-infinite-trees}, which states that the finitely many infinite trees of sorts $s \in S_{FI}$ are all built from the unique infinite trees $u_s$ with $s\in S_{1I}$.
Using this result, \cref{lem-SFI} proves the fixed point computation of $S_{FI}$ correct.
An analogous argument works for the sets $s_{\infin}$ for $s \in S_{FI}$.
\end{proof}

\begin{lemma}
\label{lem-SFF}
Let $S_{FF} \subseteq S$ be the set of sorts such that each $s \in S_{FF}$ has only finitely many finite trees.
Let $F_s^{\infin} = \{ (\con g: s_1 \times \cdots \times s_n \to s) \in F_s \mid \exists i \in \{ 1, \dots, n \}: s_i \in S_{0F} \}$ be the set of generators building only infinite trees.
Then $S_{FF}$ is the least fixed point of the following function $f: P(S) \to P(S)$:
\[
f(X) := X \cup S_{0F} \cup \{ s \in S \mid |F_s \setminus F_s^{\infin}| < \infty \land \forall (\con g: s_1 \times \cdots \times s_n \to s) \in (F_s \setminus F_s^{\infin}): \forall i \in \{ 1,\dots,n \}: s_i \in X \}
\]
where $S_{0F}$ denotes the set of sorts with only infinite trees.
\end{lemma}
\begin{proof}
Note that $f$ is monotonic, so the least fixed point is guaranteed to exist by the Knaster-Tarski theorem.
Why is $S_{FF}$ a fixed point?
First, $F_s^{\infin}$ are generators that can only construct infinite trees.
So any other generator (in $F_s \setminus F_s^{\infin}$) can construct at least one finite tree.
For there to be only finitely many finite trees in $s$, there have to be finitely many of the latter generators and for each such generator, each parameter sort $s_i$ must have only finitely many finite trees.
This explains the definition of the function $f$.

In fact, similarly to the proof of \cref{theorem-only-fin}, it is easy to see inductively that if a sort $s$ has finitely many finite trees of depth at most $d$ then $s \in f^{d+1}(\emptyset)$.
Hence $S_{FF} = \bigcup_{d \in \omega} f^{d+1}(\emptyset)$, in other words, $S_{FF}$ is in fact the least fixed point of $f$.
\end{proof}

\begin{lemma}
\label{lem-S1I}
Let $S_{1I} \subseteq S$ be the set of sorts such that each $s \in S_{1I}$ has exactly one infinite tree.
Then $S_{1I}$ is the greatest fixed point of the following function $f: P(S \setminus S_{0I}) \to P(S \setminus S_{0I})$:
\[
f(X) := \{s \in X \mid \exists (\con g: s_1 \to s) \in F_s: s_1 \in X \land (\forall (\con{g'}: s_1' \times \cdots \times s_n' \to s) \in F_s \setminus \{\con g\}: \forall i \in \{ 1,\dots,n \}: s_i' \in S_{0I}) \}
\]
Furthermore, the equations $U_s$, for $s \in S_{1I}$, that uniquely determine the unique infinite inhabitant $u_s$ of $s \in S_{1I}$ are given by the least fixed point of the function $f'$, a mapping between families of sets of equations, indexed by $s \in S_{1I}$, which for each such $s$, is given by
\[
(f'((X_s)_{s\in S_{1I}}))_{s} = \{ u_s = {\con g}_s(u_{s_1}) \} \cup X_{s_1}
\]
where ${\con g}_s: s_1 \to s$ is the unique generator with $s_1 \in S_{1I}$.
\end{lemma}
\begin{proof}
Note that $f$ is monotonic, so the greatest fixed point is guaranteed to exist by the Knaster-Tarski theorem.
Why is $S_{1I}$ a fixed point?
If $s$ has a unique infinite tree then it must start with some generator $\con g \in F_s$.
If $\con g$ had more than one parameter then the choice of the other parameter would create at least two infinite inhabitants, contradiction.
So $\con g$ has only one parameter.
Furthermore every other generator $\con{g'}$ can only create finite trees because otherwise we would lose uniqueness of the infinite tree.
The function $f$ removes all sorts from $X$ that do not satisfy these criteria.
Hence $S_{1I}$ is a fixed point.

Conversely, if a sort $s \in S \setminus S_{0I}$ has two distinct infinite trees then they have to differ at some finite depth $d$.
We claim that $s \notin f^{d+1}(S \setminus S_{0I})$ because $f$ removes sorts that have more than one infinite inhabitant and if those two inhabitants differ at depth $d$, this is detected after at most $d + 1$ applications of $f$.
This can be proved by induction, similarly to the proof of \cref{theorem-only-fin}.
Hence $S_{1I} = \bigcap_{d \in \omega} f^{d+1}(S \setminus S_{0I})$, in other words, $S_{1I}$ is the greatest fixed point of $f$.

Why is $U_s$ a fixed point of $f'$?
The equation $u_s = {\con g}_s(u_{s_1})$ must be true by the above arguments.
In order to describe $u_{s_1}$ uniquely, we need the equations $U_{s_1}$ as well.
Thus the $U_s$ are a fixed point of $f'$.
They are, in fact, the least fixed point because we are interested in the smallest set of equations describing the $u_s$.
\end{proof}

\begin{lemma}
\label{lem-shape-infinite-trees}
Let $s$ be a sort with at least one but only finitely many infinite trees.
Then each infinite tree of $s$ can be described by a term containing only variables $u_s: s$ with $s \in S_{1I}$, each representing the unique infinite tree of sort $s$.
\end{lemma}
\begin{proof}
Proof by induction on the number $\#s_{\infin}$ of infinite trees of sort $s$.
If $\#s_{\infin} = 1$ then $s$ has a unique infinite tree represented by $u_s$ and the statement is trivial.
Hence suppose $\#s_{\infin} \geq 2$.
Let $a = \con f(b_1,\dots,b_n)$ be an infinite tree with subtrees $b_1 : s_1, \dots, b_n : s_n$.

Suppose $n \ge 2$.
Then without loss of generality, assume that $b_1$ is an infinite subtree.
Since $s_2$ has at least two generators, the number of infinite inhabitants of $s_1$ is at most $\#s_{\infin} / 2 < \#s_{\infin}$.
By induction hypothesis, $b_1$ has the desired form.
The same argument works for other infinite subtrees of $a$.
For each finite subtree, there is a ground term describing it.
Hence $a$ has the desired form.

Next, suppose $n = 1$, i.e. $a = \con f(a_1)$.
If $s$ has another infinite tree $a' = \con{f'}(a_1')$ starting with a different function symbol $\con{f'}$ then the sort of $a_1$ has less than $\#s_{\infin}$ infinite trees, and the induction hypothesis gives us the desired form for $a_1$ and thus for $a$.
Otherwise, all infinite trees of $s$ start with the same function symbol $\con f$.
We can apply the same argument to $a_1$ and see that we can either proceed as above or all the infinite trees of the sort of $a_1$ must have the form $a_1 = \con f_1(a_2)$.
If we keep repeating this argument, there are three cases.

\emph{Case 1.} There are $b_1,\dots,b_n$ with $n \geq 2$ such that $a = \con f(\con f_1(\dots \con f_m(b_1,\dots,b_n) \dots ))$.
Then the first argument from above provides a term for $a$.

\emph{Case 2.} There is an $a'$ of sort $s'$ such that $a = \con f(\con f_1(\dots \con f_m(a') \dots ))$ and $s'$ has two infinite trees starting with different generators.
Then the second argument from above provides a term for $a$.

\emph{Case 3.} There is no such $a'$, meaning that the tree $a$ is uniquely determined, as an infinite path of unary function symbols.
But then $s$ only contains one infinite tree, contradiction.
So this case cannot occur.
\end{proof}

\begin{lemma}
\label{lem-SFI}
Let $S_{FI} \subseteq S$ be the set of sorts such that each $s \in S_{FI}$ has only finitely many infinite trees.
Let $F_s^{\infin} = \{ (\con g: s_1 \times \cdots \times s_n \to s) \mid \exists i \in \{1,\dots,n\}: s_i \notin S_{0I} \}$ be the set of generators that can construct infinite trees.
Then $S_{FI}$ is the least fixed point of the following function $f: P(S) \to P(S)$:
\begin{align*} f(X) := X &\cup S_{0I} \cup S_{1I} \cup \{ s \in S \mid |F_s^{\infin}| < \infty \land \forall \con g: s_1 \times \cdots \times s_n \to s \in F_s: \\
&\forall i \in \{1,\dots,n\}: s_i \in S_{0I} \lor (s_i \in X \land (\forall j \in \{1,\dots,n\} \setminus \{i\}: s_j \in S_{FF} \cap X)) \}.
\end{align*}
\end{lemma}
\begin{proof}
Note that $f$ is monotonic, so the least fixed point is guaranteed to exist by the Knaster-Tarski theorem.
Why is $S_{FI}$ a fixed point?
First of all, it is clear that $S_{0I} \cup S_{1I} \subseteq S_{FI}$.
Furthermore, for the sort $s$ to have finitely many infinite trees, there have to be finitely many generators $F_s^{\infin}$ that can construct infinite trees.
Additionally, for each such generator $\con g: s_1 \times \cdots \times s_n \to s \in F_s^{\infin}$, there have to be finitely many infinite trees starting with $\con g$.

It is easier to describe the negation of this: If a generator $\con g$ starts infinitely many infinite trees, there must be a parameter $i$ such that $s_i$ contains infinite trees and one of the following: (1) $s_i$ containing infinitely many infinite trees or (2) one of the other $s_j$ containing infinitely many trees.
In either case, this leads to infinitely many infinite trees starting with $\con g$. This can be formulated as
\[ \exists i \in \{1,\dots,n\}: s_i \notin S_{0I} \land (s_i \notin S_{FI} \lor \exists j \in \{1,\dots,n\} \setminus \{i\}: s_j \notin S_{FF} \cap S_{FI}). \]
The negation of this is what is written in the above function definition.
Hence $S_{FI}$ is a fixed point of $f$.

Next, we show that $S_{FI}$ is the least fixed point.
Let $s$ be a sort with finitely many infinite trees and $a$ such a tree.
By \cref{lem-shape-infinite-trees}, there is a term $t_a$ describing $a$, containing only variables $u_i: s_i$ representing the unique infinite tree of $s_i$.
We always choose $t_a$ to be of minimal depth among those terms.
By definition of $f$, each $s_i \in f(\emptyset)$.
Let $a$ be the infinite tree in $S_{FI}$ such that its corresponding $t_a$ has maximal depth $d$.
Then one can see inductively, as in the proof of \cref{theorem-only-fin}, that $s \in f^{d+1}(\emptyset)$.
Hence $S_{FI} = \bigcup_{d \in \omega} f^{d+1}(\emptyset)$, in other words, $S_{FI}$ is in fact the least fixed point of $f$.
\end{proof}

\section{Supplementary material for Section \ref{sec-solver}}

Throughout this section, we assume the variable convention that bound variables of terms occurring in a certain mathematical context (like definitions and proofs) are assumed to be distinct and different from the free variables.
Furthermore, when talking about reachability in a formula $\exists \bar x. \alpha$ where $\alpha$ is a basic formula, we mean reachability in $\alpha$ from the free variables of the whole formula.

We are also going to need the \emph{Unique Solution Axiom} \cite[Axiom 3 in Section 3.2]{ddf08}, which states that for any sequence of distinct variables $\bar z$ and non-variable terms $t_i$ containing only the variables $\bar x$ and $\bar z$, we have
\[ \forall \bar x \ldotp \exists! \bar z \ldotp \bigwedge_i z_i = t_i \]
in the extended theory of trees.
This is proved in \cite[Theorem 3.3.1]{ddf08}.

\begin{algorithm}
\caption{Algorithm for solving a basic formula $\alpha$ with free variables $v_0 < \dots < v_n$.
The rules 1--10 are taken from \cite{ddf08}.
The two rules in blue at the end are new.}
\label{alg-solve-basic}
\begin{algorithmic}
\Function{solveBasic}{$(v_0, \dots, v_n), \alpha$}
\State let $<$ be the ordering where $v_0 < \cdots < v_n$
\Repeat
\If{$\alpha$ is $u = u \land \alpha'$} $\alpha \gets \alpha'$ \Comment{Rule 1 (numbering as in \cite[Section 4.6]{ddf08})} \EndIf
\If{$\alpha$ is $u = v \land \alpha'$ and $u < v$} $\alpha \gets v = u \land \alpha'$ \Comment{Rule 2} \EndIf
\If{$\alpha$ is $v = u \land v = t \land \alpha'$ and $u < v$} $\alpha \gets v = u \land u = t \land \alpha'$ \Comment{Rule 3} \EndIf
\If{$\alpha$ is $u = \con f(\bar y) \land u = \con g(\bar z) \land \alpha'$ and $\con f \not\equiv \con g$} \Return false \Comment{Rule 4} \EndIf
\If{$\alpha$ is $u = \con f(\bar y) \land u = \con f(\bar z) \land \alpha'$} $\alpha \gets u = \con f(\bar y) \land \overline{y = z} \land \alpha'$ \Comment{Rule 5} \EndIf
\Until{no changes in the last iteration}
\Repeat
\If{$\alpha$ is $\fin(u) \land \fin(u) \land \alpha'$} $\alpha \gets \fin(u) \land \alpha'$ \Comment{Rule 7} \EndIf
\If{$\alpha$ is $v = u \land \fin(v) \land \alpha'$ and $u < v$} $\alpha \gets v = u \land \fin(u) \land \alpha'$ \Comment{Rule 8} \EndIf
\If{$\alpha$ is $\fin(u) \land \alpha'$ and $u$ is properly reachable from $u$} \Return false \Comment{Rule 9} \EndIf
\If{$\alpha$ is $u = \con f(\bar y) \land \fin(u) \land \alpha'$} $\alpha \gets u = \con f(\bar y) \land \overline{\fin(y)} \land \alpha'$ \Comment{Rule 10} \EndIf
\color{blue}
\If{$\alpha$ is $\fin(u) \land \alpha'$ and $u:s$ with $s \in S_{0I}$} $\alpha \gets \alpha'$ \Comment{(*)} \EndIf
\If{$\alpha$ is $\fin(u) \land \alpha'$ and $u:s$ with $s \in S_{0F}$} \Return false \Comment{(*)} \EndIf
\color{black}
\Until{no changes in the last iteration}
\State \Return $\alpha$
\EndFunction
\end{algorithmic}
\end{algorithm}

\begin{theorem*}[\cref{thm-solve-basic}, repeated]
The function $\Call{solveBasic}{(v_0, \cdots, v_n), \alpha}$ from \cref{alg-solve-basic} correctly solves basic formulae $\alpha$ containing the variables $v_0, \cdots, v_n$, i.e.\ it turns $\alpha$ into an equivalent solved formula (with respect to the variable ordering $v_0 < \cdots < v_n$) or returns $\false$ if none exists.
\end{theorem*}
\begin{proof}[Proof of \cref{thm-solve-basic}]
Most parts of the algorithm (the numbered rules) are taken from \cite[Section 4.6]{ddf08} and the fact that the result satisfies property (1) of \cref{def-basic} is proven there.
(Note that Rule 6 is used only for bookkeeping in \cite{ddf08}, which is why it is not needed in our formulation of the algorithm.)
The two additional rules (*) involve variables $u:s$ where $s$ is a sort without infinite, respectively finite, trees.
Obviously, $\fin(u)$ is always, respectively never, satisfied in those cases.
Therefore, the result satisfies property (2) of \cref{def-basic} as well.
\end{proof}

\begin{theorem*}[\cref{thm-instantiations}, repeated]
Let $\phi \equiv \lnot(\exists \bar x \ldotp \alpha \land \bigwedge_i \phi_i)$ be a normal formula of depth at most 2 with free variables $\bar v$.
Let $I$ be the result of $\Call{findInstantiation}{\bar v, \phi}$ from \cref{alg-find-instantiations}.
If $I$ is ``none", then there is no instantiable variable.
Otherwise, let $u$ be the first instantiable variable found in \Call{findInstantiation}{}.
Then $\phi$ is equivalent to the following conjunction of normal formulae, in which the variable $u$ is no longer instantiable:
\[ \bigwedge_{(\exists \bar z \ldotp \psi) \in I} \lnot(\exists \bar x\bar z \ldotp \alpha \land \psi \land \bigwedge_i \phi_i). \]
\end{theorem*}
\begin{proof}[Proof of \cref{thm-instantiations}]
We first show that in each case, the result $I$ of the call to \Call{findInstantiation}{} satisfies $\alpha \to \bigvee_{\psi \in I} \psi$.
For the first return statement in \Call{findInstantiation}{}, this is clear because if there are finitely many generators of $s$ then one of them has to be used to construct a tree of sort $s$.
For the second return statement, it is clear because $u$ can have only finitely many values, so if $I$ contains formulae describing each possible value then the disjunction over all of them must be true.
For the third return statement, note that $\fin(u)$ occurs in $\alpha$, so $u$ has to be finite.
Hence $I$ only contains formulae describing each finite value of $s$, and we have $\fin(u) \to \bigvee_{\psi \in I} \psi$.
Finally, consider the fourth return statement.
The variable $u$ has to represent either a finite tree, meaning $\fin(u)$ or one of the finitely many infinite trees in $s$.
Again, we find that $\bigvee_{\psi \in I} \psi$ holds.
Since $\alpha \to \bigvee_{\psi \in I} \psi$ holds in each case, $\phi$ is equivalent to:
\allowdisplaybreaks
\begin{align*}
&\lnot\left(\exists \bar x \ldotp \alpha \land \left(\bigvee_{\psi \in I} \psi\right) \land \bigwedge_i \lnot(\exists \bar y_i \ldotp \beta_i)\right) \\
&\leftrightarrow \lnot\left(\exists \bar x \ldotp \alpha \land \left(\bigvee_{(\exists \bar z \ldotp \psi') \in I} (\exists \bar z \ldotp \psi')\right) \land \bigwedge_i \lnot(\exists \bar y_i \ldotp \beta_i)\right) \\
&\leftrightarrow \lnot\left(\bigvee_{(\exists \bar z \ldotp \psi') \in I} \exists \bar x\bar z \ldotp \alpha \land \psi' \land \bigwedge_i \lnot(\exists \bar y_i \ldotp \beta_i)\right) \\
&\leftrightarrow \bigwedge_{(\exists \bar z \ldotp \psi') \in I} \lnot\left(\exists \bar x\bar z \ldotp \alpha \land \psi' \land \bigwedge_i \lnot(\exists \bar y_i \ldotp \beta_i)\right)
\end{align*}
By the construction of $I$, the instantiable variable $u$ found in the algorithm is no longer instantiable in this transformed formula.
\end{proof}

In the following, we will often have to show that for sorts $s$ with infinitely many trees of sort $s$, there is a tree that contradicts a certain set of finitely many equations.
The following lemma formalizes this.

\begin{lemma}
\label{lem-contradict-recursive-constraints}
Let $v$ be a variable of sort $s \notin S_{FI} \cap S_{FF}$.
Let $\beta_i$ be a family of solved basic formula, indexed by $i = 1, \dots, m$, where $v$ is properly reachable from itself in each $\beta_i$.
Let $T$ be a finite set of trees of sort $s$.
Then there is a tree $v^*$ of sort $s$ such that $v^* \notin T$ and each $\beta_i$ is made false by any valuation with value $v^*$ for $v$.
(Roughly speaking, $\beta_i$ are ``forbidden recursive equations", $T$ are ``forbidden values" and $v^*$ avoids both.)
Furthermore, $v^*$ can be chosen to be finite if $s \notin S_{FF}$ and infinite if $s \notin S_{FI}$.
\end{lemma}
\begin{proof}
If $s \notin S_{FF}$ then there is a finite tree $v^*$ not in the finite set $T$.
This value for $v$ also makes each $\beta_i$ false because it is finite but $v$ is properly reachable from itself in $\beta_i$ and thus infinite.

For the other case, $s \notin S_{FI}$, we need the concept of \emph{contradicting a formula at a certain depth}.
Let $\beta$ be a basic formula containing a subformula of the form
\[ x_0 = \con{g}_0(\dots, x_1, \dots) \land x_1 = \con{g}_1(\dots, x_2, \dots) \land \cdots x_{n-1} = \con{g}_{n-1}(\dots, x_0, \dots). \]
Let $\mathcal V$ be a valuation of the variables of $\beta$ that maps $x_i$ to a tree $t$ and let $\bar y, \bar z$ be variables such that the equation with left-hand side $x_i$ in $\beta$ is $x_i = \con{g}_i(\bar y, x_{i+1}, \bar z)$.
(We view the indices of $x$ modulo $n$.)
Next, we define
\newcommand{\contra}{\mathrm{contra}}
\[
\contra(x_i, t, \beta, \mathcal V) = \begin{cases}
0 &\text{if $\con{g}_i$ is not the root of $t$} \\
\contra(x_{i+1}, t', \beta) + 1 &\text{if $t'$ is the subtree of $t$ with $t = \con{g}_i^\mathcal T(\overline{\mathcal V(y)}, t', \overline{\mathcal V(z)})$} \\
\infty &\text{otherwise}
\end{cases}
\]
and say that \emph{the value $t$ for $x_i$ contradicts $\beta$ at depth $\contra(x_i, t, \beta, \mathcal V)$ under the valuation $\mathcal V$}.
Intuitively, this means that when picking the value $t$ for $x_1$ and checking the equations in $\beta$, we notice a problem at depth $d$ of the tree.
If no valuation is specified, we define
\[ \contra(x_i, t, \beta) := \max_{\mathcal V} \contra(x_i, t, \beta, \mathcal V) \]
where $\mathcal V$ ranges over valuations sending $x_i$ to $t$, and say that \emph{the value $t$ for $x_i$ contradicts $\beta$ at depth $\contra(x_i, t , \beta)$}.
If $\contra(x_i, t, \beta) < \infty$, we say that \emph{the value $t$ for $x_i$ contradicts $\beta$}.
Note that only the nodes up to depth $d$ are relevant for contradicting $\beta$ at depth $d$.

Phrased in this new terminology, our goal is to prove that there is a tree $v^* \notin T$ of sort $s$ such that the value $v^*$ for $v$ contradicts $\beta_1,\dots,\beta_m$.
In the following, we will iteratively construct a sequence of injections $f_0,\dots,f_m: s^\mathcal T \to s^\mathcal T$, and of infinite sets $W_0, \dots, W_m \subseteq s^\mathcal T$, such that for each $i \in \{1,\dots,m\}$ and all $w \in W_i$, the value $f_i(w)$ for $v$ contradicts all the formulae $\beta_1,\dots,\beta_i$.

The base case is easy, simply define $f_0$ as the identity function, and $W_0$ as the set of infinite trees of sort $s$.
Next, suppose $f_i$ and $W_i$ are defined with the desired property.
If the value $f_i(w)$ for $v$ contradicts $\beta_{i+1}$ for all $w \in W_i$, we can simply use $W_{i+1} = W_i$, and $f_{i+1} = f_i$.
Otherwise there is a $w \in W_i$ such that there is a valuation $\mathcal V$ sending $v$ to $f_i(w)$ that makes $\beta_{i+1}$ true.
Since $v$ is reachable from itself, $\beta_{i+1}$ contains a subformula
\[ v = \con{g}_0(\dots, x_1, \dots) \land x_1 = \con{g}_1(\dots, x_2, \dots) \land \cdots x_{n-1} = \con{g}_{n-1}(\dots, v, \dots). \]
Let $t = f_i(w)$.
We are going to label the nodes of $t$ with the corresponding variables of $\beta_{i+1}$.
That is to say, we label the root of $t$ with $v$, the child node corresponding to $x_1$ with $x_1$, and so on, such that the labeled nodes form an infinite path labeled $v,x_1,\dots,x_{n-1},v,x_1,\dots$ in the tree $t$.
Let $d$ be an integer such that the value $t$ for $v$ contradicts each $\beta_1,\dots,\beta_i$ at depth at most $d$.
It exists because each the value $t$ for $v$ contradicts each $\beta_1,\dots,\beta_i$ at some finite depth by the induction hypothesis.
Let $n$ be a node in $t$ labeled $v$, at a depth $> d$.
The subtree rooted at $n$ must be $t$ again because otherwise, $\beta_{i+1}$ would not be true.
If we replace this subtree by a different subtree, $\beta_{i+1}$ cannot be satisfied under any valuation of the variables $x_1, \dots, x_n$ because the values of the latter are determined by other subtrees of $t$.
Hence let $W_{i+1} = W_i \setminus \{ t \}$ (which is also infinite) and $f_{i+1}(w)$ be the function returning $t$ but with the subtree rooted at $n$ replaced by $w$.
By construction, the value $f_{i+1}(w)$ for $v$ contradicts $\beta_{i+1}$.
Since for each $w \in W_{i+1}$, all nodes of $t$ and $f_{i+1}(w)$ agree up to depth $d$, the value $f_{i+1}(w)$ for $v$ also contradicts $\beta_1,\dots,\beta_i$ at depth at most $d$, as desired.

At the end of this iterative process, we obtain an infinite set $W_m$ and an injection $f_m: s^\mathcal T \to s^\mathcal T$ such that for all $w \in W_m$, the value $f_m(w)$ for $v$ contradicts $\beta_1,\dots,\beta_m$.
The set $W := \{ f_m(w) \mid w \in W_i \}$ is infinite because $f_m$ is injective.
Furthermore, each tree $t \in W$ contradicts all $\beta_i$.
Since $T$ is finite, there is a tree $v^* \in W \setminus T$.
\end{proof}

\begin{theorem*}[\cref{thm-solved-form}, repeated]
Let $\phi$ be a fully simplified formula.
If $\phi$ has no free variables then $\phi \equiv \true$.
Otherwise both $\phi$ and $\lnot \phi$ are satisfiable in the theory of trees.
\end{theorem*}
\begin{proof}[Proof of \cref{thm-solved-form}]
The formula $\phi$ has the form
\[ \exists \bar x \ldotp \alpha \land \bigwedge_{i \in I} \lnot(\exists \bar y_i \ldotp \beta_i). \]
First consider the case of no free variables.
Then no variable can be reachable in $\exists \bar x \ldotp \alpha$, hence by the reachability condition (5) of \cref{def-solved}, $\bar x$ is empty.
This implies that $\alpha$ is just $\true$ because it cannot mention any variables.
The same argument applied to each $\exists y_i \ldotp \beta_i$ means that $\bar y_i$ is empty and $\beta_i \equiv \true$.
By condition (3) of \cref{def-solved}, each $\beta_i$ must include a conjunct not occurring in $\alpha$, hence $I = \emptyset$.
Altogether, we have $\phi \equiv \true$

If $\phi$ contains free variables, it is enough to find a valuation for the free variables such that $\phi$ is true in the theory of trees and another one such that $\phi$ is false in the theory of trees.
To find a valuation that makes $\phi$ false, consider the following:
If $\alpha$ contains a free variable $z$, it can be made false like this.
\begin{itemize}[noitemsep, nolistsep]
\item If $\alpha$ contains $z = w$, then $z > w$ according to the variable ordering since $\alpha$ is solved.
Hence $w$ is also a free variable and $\alpha$ can be made false by instantiating $z$ and $w$ with different trees.
\item If $z = \con f(\bar w)$ occurs in $\alpha$, it is enough to instantiate $z$ with a tree not starting with $\con f$ to make $\alpha$ false, which is always possible because each sort has at least two generators.
\item If $w = t$ with $t$ containing $z$ occurs in $\alpha$, this equation must be reachable in $\exists \bar x \ldotp \alpha$ by condition (5) of \cref{def-solved}.
This means that there is an equation of the form $z' = \dots$ in $\alpha$, with $z'$ free and $w$ reachable from $z'$.
This situation was already handled in one of the previous two cases.
\item If $\fin(z)$ occurs in $\alpha$, simply instantiate $z$ to an infinite tree (which is possible by condition (2) of solved basic formulae) to make $\alpha$ false.
\end{itemize}
Otherwise, $\alpha$ contains no free variables, so $\bar x$ is empty and $\alpha$ is $\true$ by the same argument as before.
Since $\phi$ contains a free variable, there must be a $\beta_i$ that contains a free variable, so is nonempty.
Since $\beta_i$ is a solved basic formula, it is satisfiable by \cref{lem-solved-basic-satisfiable}.
Hence there is a valuation of free variables that makes $\lnot\exists \bar y_i \ldotp \beta_i$ false.
Then the same valuation makes $\phi$ false.

Next, we want to find a valuation making $\phi$ true.
Let $\beta_i^*$ be $\beta_i$ with all conjuncts occurring in $\alpha$ removed.
Our goal is to find a valuation of the free variables and $\bar x$ that makes $\alpha$ true and every $\exists \bar y_i \ldotp \beta_i^*$ false (since we cannot make the parts of $\beta_i$ that also occur in $\alpha$ false).
Let $\bar x_{lhs}$ me the variables from $\bar x$ that occur on the left-hand side of an equation in $\alpha$.
The valuation for these variables will be picked last because it is uniquely determined by the Unique Solutions Axiom, once the valuation for the other variables is chosen.
So the equations of $\alpha$ are taken care of.

If $\fin(v)$ occurs in $\alpha$ then any equation $v = \con f(\bar w)$ occurring in any $\beta_i^*$ is automatically false because $v$ has to be properly reachable from itself (otherwise $v$ would be instantiable), but then $v$ cannot be finite.
So the only equations with $v$ on the left-hand side in $\beta_i$ that we care about are $v = w_i$ for other variables $w_i$.
In this case, each $w_i$ is also a free variable because $v > w_i$ by the variable ordering, and the sort of $v$ has infinitely many finite trees because otherwise $v$ would be instantiable.
Thus it is always possible to find a valuation that contradicts all these finitely many equations of the form $v = w_i$ by picking a value for $v$ that is different from the values picked for all the $w_i$.
This proves that we can always make $\alpha$ true.

Next, we do a case analysis on the $\beta_i^*$ that have not been made false yet.
By reachability, each $\beta_i^*$ has to contain $\fin(v)$ or $v = t$ for a free variable $v$.
Then $v \notin \bar x_{lhs}$ because $\beta_i$ is solved.
We can assume that $\fin(v)$ does not occur in $\alpha$ because this case was already discussed above.
For each such free variable $v$, we do the following case analysis:
\begin{itemize}[noitemsep, nolistsep]
\item Suppose there is a $\beta_i^*$ that contains $\fin(v)$.
If $\beta_i^*$ also contains an equation, then the following cases apply and suffice to make it false.
So suppose $\beta_i^*$ only contains $\fin$-constraints.
Then the sort of $v$ has infinitely many infinite trees because otherwise $v$ would be instantiable.
This makes the following cases work, by restricting the set of possible values for $v$ to the set of infinite trees.
Using such a value also makes $\fin(v)$, and thus $\beta_i^*$, false as desired.
\item Suppose $v = \con f(\bar w)$ occurs in some $\beta_i^*$.
Then $v$ must be properly reachable from itself in $\beta_i$ because otherwise, it would be instantiable.
If the sort of $v$ had only finitely many trees then $v$ would be instantiable, contradiction.
Hence the sort of $v$ has infinitely many trees.
Since the previous cases are already handled, we can assume that the only constraints on $v$ in all the $\beta_j^*$'s are of the form $v = \con f(\bar w)$ with $v$ properly reachable from itself in $\beta_j^*$ or $v = w$.
Since the sort of $v$ has infinitely many trees, it is possible to contradict all these constraints by \cref{lem-contradict-recursive-constraints}
\item Suppose $v = w$ occurs in some $\beta_i$.
Since the previous cases are already handled, we can assume that the only constraints on $v$ from the $\beta_i^*$'s are of the form $v = w_i$ for variables $w_i$.
Then each $w_i$ is also a free variable because $v > w_i$ by the variable ordering, and the sort of $v$ has infinitely many trees because otherwise $v$ would be instantiable.
Thus it is always possible to find a valuation that contradicts all these finitely many equations of the form $v = w_i$ by picking a value for $v$ that is different from the values picked for all the $w_i$.
\end{itemize}
This case analysis shows that we can make all the $\beta_i^*$ false.
Thus it is always possible to find a valuation that makes $\phi$ true, as desired.
\end{proof}

\begin{lemma}
\label{lem-solved-basic-satisfiable}
Any solved basic formula is satisfiable.
\end{lemma}
\begin{proof}
Let the basic formula be given by $\overline{v = t} \land \overline{\fin(u)}$.
By condition (2) of solved basic formulae (\cref{def-basic}), each $\bar u$ can be given the value of some finite tree.
Since the variables $\bar v$ and $\bar u$ are disjoint, the Unique Solution Axiom tells us that $\exists \bar v \ldotp \overline{v = t}$ is satisfiable for this valuation of $\bar u$.
\end{proof}

\begin{algorithm}
\caption{Extension of Djelloul, Dao, and Frühwirth's algorithm \cite[Section 4.6]{ddf08} for transforming a normal formula into an equivalent conjunction of solved formulae. The added part is in blue.}
\label{alg-solve-normal}
\begin{algorithmic}
\color{black}
\Function{solve}{$\phi$}
\State $\tilde \phi \gets \Call{normalize}{\lnot \phi}$ \Comment{cf. \cref{thm-normalize}}
\State $\bar v \gets $ the free variables of $\tilde \phi$ in some fixed order
\State let $\lnot(\exists \bar x \ldotp \alpha \land \bigwedge_i \phi_i) \equiv \tilde \phi$ where $\alpha$ is a basic formula and $\phi_i$ are normal formulae
\State $\alpha \gets \Call{solveBasic}{\bar v \bar x, \alpha}$ \Comment{cf. \cref{alg-solve-basic}}
\If{$\alpha \equiv \false$} \Return $\false$ \EndIf
\State $\{ \psi_1,\dots,\psi_n \} \gets \Call{solveNested}{\bar v, \lnot(\exists \bar x \ldotp \alpha \land \bigwedge_i \phi_i)}$
\If{$n = 0$} \Return $\false$ \EndIf
\If{each $\psi_i$ is of the form $\lnot(\true)$} \Return $\true$ \EndIf
\State let $\lnot(\exists \bar x_i \ldotp \alpha_i \land \bigwedge_{j\in J_i} \lnot(\exists \bar y_{ij} \ldotp \beta_{ij})) \equiv \psi_i$ for each $i$
\State remove all conjuncts (of the form $u = t$ or $\fin(u)$) from each $\beta_{ij}$ that already occur in $\alpha_i$
\State \Return $\bigvee_{i=1}^n \big(\exists \bar x_i \ldotp \alpha_i \land \bigwedge_{j\in J_i} \lnot(\exists \bar y_{ij} \ldotp \beta_{ij}) \big)$
\EndFunction

\Function{solveNested}{$\bar v, \phi$}
\State let $\lnot(\exists \bar x \ldotp \alpha \land \bigwedge_i \lnot(\exists \bar y_i \ldotp \phi_i)) \equiv \phi$
\For{each $i$}
\State let $\beta_i \land \bigwedge_j \psi_{ij} \equiv \phi_i$ where $\beta_i$ is a basic formula and $\psi_{ij}$ are normal formulae
\State $\beta_i \gets \alpha \land \beta_i$ \Comment Rule 12 in \cite[Section 4.6]{ddf08}
\State $\beta_i \gets \Call{solveBasic}{\bar v\bar x\bar y_i, \beta_i}$
\If{$\beta_i \equiv \false$} $\Psi_i \gets \emptyset$
\Else
\State replace each $u = t$ in $\beta_i$ by $u = s$ if $u = s$ occurs in $\alpha$ \Comment Rule 13
\State $\Psi_i \gets \Call{solveNested}{\bar v\bar x, \lnot(\exists \bar y_i \ldotp \beta_i \land \bigwedge_j \psi_{ij})}$
\EndIf
\EndFor
\State \Return $\Call{solveFinal}{\bar v, \lnot(\exists \bar x \ldotp \alpha \land \bigwedge (\bigcup_i \Psi_i))}$
\EndFunction

\Function{solveFinal}{$\bar v, \phi$}
\State let $\lnot(\exists \bar x \ldotp \alpha \land \bigwedge_i \phi_i) \equiv \phi$
\If{there is an $i$ such that $\phi_i \equiv \lnot (\exists \bar y \ldotp \alpha)$}
\Return $\emptyset$ \Comment Rule 14
\EndIf
\If{depth of $\phi$ is 3} \Comment Rule 16 (depth reduction)
\State choose a $j$ such that $\phi_j$ has depth 2
\State let $\lnot(\exists \bar y \ldotp \beta \land \bigwedge_k \lnot(\exists \bar z_k \ldotp \gamma_k)) \equiv \phi_j$
\State $\psi \gets \lnot(\exists \bar x \ldotp \alpha \land \lnot(\exists \bar y \ldotp \beta) \land \bigwedge_{i, i\neq j} \phi_i)$
\State $\chi_k \gets \lnot(\exists \bar x\bar y\bar z_i \ldotp \gamma_i \land \bigwedge_{i,i\neq j} \phi_i)$
\State \Return $\Call{solveFinal}{\bar v, \psi} \cup \bigcup_k \Call{solveNested}{\bar v, \chi_k}$
\EndIf
\color{blue}
\State $I \gets \Call{findInstantiation}{\bar v, \phi}$ \Comment{cf. \cref{alg-find-instantiations}}
\If{$I \neq none$}
\State \Return $\bigcup \{ \Call{solveNested}{\bar v, \lnot(\exists \bar x\bar z \ldotp \alpha \land \psi \land \bigwedge_i \phi_i)} \mid (\exists \bar z \ldotp \psi) \in I \}$ \Comment{cf. \cref{thm-instantiations}}
\Else \color{black}
\State \Return $\{ \Call{removeUnreachableParts}{\bar v, \phi} \}$ \Comment{Rule 15 (cf. \cref{alg-remove-unreachable})}
\EndIf
\EndFunction
\end{algorithmic}
\end{algorithm}

\begin{algorithm}
\caption{Rule 15 from \cite[Section 4.6]{ddf08}, which removes unreachable variables and subformulae of a normal formula of depth at most 2.}
\label{alg-remove-unreachable}
\begin{algorithmic}
\Function{removeUnreachableParts}{$\bar v, \lnot(\exists\bar x \ldotp \alpha \land \bigwedge_i \psi_i)$}
\Comment Rule 15
\State $\bar x' \gets$ the variables from $\bar x$ reachable in $\alpha$ from the free variables $\bar v$
\State $\alpha' \gets $the conjuncts of $\exists \bar x \ldotp \alpha$ reachable from the free variables $\bar v$
\State $\alpha'' \gets$ the $\fin$-subformulae of $\exists \bar x \ldotp \alpha$  unreachable from the free variables $\bar v$
\State $\alpha''' \gets$ the equations of $\exists \bar x \ldotp \alpha$ unreachable from the free variables $\bar v$
\State $\bar x''' \gets $ the variables of $\exists \bar x \ldotp \alpha$ occurring on the LHS of an equation in $\alpha$, and unreachable from $\bar v$
\State $\bar x'' \gets \bar x$ without $\bar x'$ and $\bar x'''$
\For{$i$}
    \State Let $\lnot(\exists y_i \ldotp \beta_i) \equiv \phi_i$
    \State $\beta_i^* \gets \beta_i$ with $\alpha''$ removed
    \State $\bar y_i' \gets $ the variables of $\bar x'''\bar y_i$ in $\exists \bar x'''\bar y_i \ldotp \beta_i^*$ reachable from its free variables
    \State $\beta_i' \gets$ the conjuncts of $\exists \bar x'''\bar y_i \ldotp \beta_i^*$ reachable from its free variables
\EndFor
\State $K \gets$ the set of indices $i$ where no variable of $\bar x''$ occurs in $\beta_i'$
\State \Return $\{\lnot\exists \bar x' \ldotp \alpha' \land \bigwedge_{i\in K} \lnot(\exists \bar y_i' \ldotp \beta_i')\}$
\EndFunction
\end{algorithmic}
\end{algorithm}

\begin{theorem*}[\cref{thm-solver-correct}, repeated]
Given a formula $\phi$, the function $\Call{solve}{\phi}$ from \cref{alg-solve-normal} returns $\true$, $\false$, or a disjunction of fully simplified formulae that is equivalent to $\phi$ in the extended theory of trees.
In particular, if $\phi$ is closed, it returns $\true$ or $\false$.
\end{theorem*}
\begin{proof}[Proof of \cref{thm-solver-correct}]
The proof of this is quite involved and will take up the rest of this section.
The function \Call{solve}{$\phi$} first normalizes $\lnot \phi$ and then solves its basic formula.
If the latter contains a contradiction, $\phi$ is unsatisfiable.
Otherwise, the function \Call{solveNested}{} recursively solves $\lnot \phi$: it returns a set of solved normal formulae $\{ \psi_1,\dots,\psi_n \}$ such that $\lnot \phi$ is equivalent to $\bigwedge_{i=1}^n \psi_i$.
It works very similarly to the original algorithm in \cite[rules 12--16 in Section 4.6]{ddf08}.
(Note that Rule 11 is used only for bookkeeping in \cite{ddf08}, which is why it is not needed in our formulation of the algorithm.)
The only change is the instantiation step, highlighted in \cref{alg-solve-normal}.
The unchanged parts are proven correct in \cite[Property 4.6.3]{ddf08}.

The following lemmas prove the correctness of our change.
\Cref{lem-termination} establishes the termination of repeated instantiation steps.
The termination of the unchanged parts of the original algorithm is shown in \cite[Property 4.6.3]{ddf08}.
The fact that the instantiation step is correct was proven in \cref{thm-instantiations} already.
Next, \Cref{lem-conditions1-4} proves that the properties (1) to (4) of a solved formula (\cref{def-solved}) are satisfied when \Call{removeUnreachableParts}{} is called.
\Cref{lem-rule15-correct} proves that the return value of \Call{removeUnreachableParts}{} is correct.
By construction, it satisfies property (5) as well, thus it is solved.

Since the return value of \Call{solveNested}{} is a set of solved normal formulae $\{ \psi_1,\dots,\psi_n \}$ such that $\lnot \phi$ is equivalent to $\bigwedge_{i=1}^n \psi_i$, the original formula $\phi$ is equivalent to $\bigvee_{i=1}^n \lnot\psi_i$, a disjunction of fully simplified formulae.
In particular, if $n = 0$ then $\phi$ is always false.
Conversely, if each $\psi_i$ is $\lnot\true$ then $\phi$ is always true.
In all other cases, we remove subformulae that were duplicated by Rule 12 in \Call{solveNested}{}.
This last simplification step is not strictly necessary: even without it, the results would be fully simplified formulae.
Finally, we return the whole disjunction.
\end{proof}

In order to prove the termination of repeated instantiations, we need the following concept.

\begin{definition}[depth of a variable]
Let $\alpha$ be a solved basic formula.
The \emph{depth} of a variable $v$ in $\alpha$, denoted by $\depth_\alpha(v)$, is defined as follows.
If $v$ is properly reachable from itself or doesn't occur on the left-hand side of an equation in $\alpha$, its depth is 0.
Else if $v = \con f(\bar w)$ occurs in $\alpha$, its depth is $\depth_\alpha(v) := 1 + \max_i(\depth_\alpha(w_i))$.
Else if $v = w$ occurs in $\alpha$, its depth is $\depth_\alpha(v) := \depth_\alpha(w)$.
\end{definition}

Note that this is well-defined because of the ``cycle check" using reachability in the definition.

\begin{lemma}
\label{lem-termination}
There are only finitely many instantiations (calls to \Call{findInstantiation}{} that do not return ``$none$") happening in \cref{alg-solve-normal}.
Hence the algorithm terminates.
\end{lemma}
\begin{proof}
For a given normal formula $\phi \equiv \lnot(\exists \bar x \ldotp \alpha \land \bigwedge_i \lnot(\exists \bar y_i \ldotp \beta_i))$ of depth 2 with free variables $\bar v$, let $X$ be the set of instantiable variables and define
\[ N_j(\phi) = |\{ v \in X \mid \max_i(\depth_{\beta_i}(v)) = j \}|, \]
in other words, the number of instantiable variables with maximum depth $j$.
Let $k$ be the maximum integer such that $N_k(\phi) > 0$ and define
\[ N(\phi) = (N_k(\phi),\dots,N_1(\phi),N_0(\phi)). \]

By the definition of \Call{findInstantiation}{}, the Instantiation Rule is only applied if $N(\phi) \neq (0,\dots,0)$.
We claim that the value of $N(\phi)$ decreases with respect to lexicographical order in each recursive call of \Call{solveNormalized}{} after every application of the Instantiation Rule.
Note that it was proved in \cite[Property 4.6.3]{ddf08} that when \Call{findInstantiation}{} is called, the normal formula $\phi$ satisfies conditions (1) to (3) of \cref{def-solved}.

Suppose the variable $u$ returned by \Call{findInstantiation}{} was selected because there is an equation $u = \con f(\bar w)$ in $\beta_i^*$ where $u$ is not properly reachable from $u$.
Then $u$ does not occur on a LHS in $\alpha$ because of the variable ordering: If $u = v$ occurred in $\alpha$, it would also occur in $\beta_i$ by condition (2) of solved basic formulae (\cref{def-basic}) and $\beta_i$ would not be solved, violating condition (1).
After instantiating $u$, i.e. adding the equation $u = \con g(\bar z)$, the resulting basic formula $\alpha \land u = \con g(\bar z)$ is therefore solved, so \Call{solveBasic}{} does not change it at all.
Next, Rule 12 copies $\alpha$ into each $\beta_j$.
What can happen in \cref{alg-solve-basic} now?
If $\beta_i$ contains $u = \con f(\bar w)$, this leads to the situation $u = \con g(\bar z) \land u = \con f(\bar w)$.
If $\con f \not\equiv \con g$, this is a conflict and $\beta_i$ is removed from $\phi$.
Otherwise, that part of $\beta_i$ is replaced with $\overline{z = w}$.
Given that $\beta_i$ was a solved basic formula before adding $\overline{z = w}$, the only applicable rule in $\beta_i$ is Rule 2, switching the ordering of $z_k = w_k$ to $w_k = z_k$ if $w_k > z_k$.
If there is another equation $w_k = t$, Rule 3 will change it to $w_k = z_k$ and $z_k = t$.
Afterward, no more rules are applicable, and the resulting formula is solved.
(In fact, Rule 8 or 10 could also be applied but this is irrelevant for the depths.)
Denote the formulae resulting from the original $\phi$ and $\beta_i$ by $\phi'$ and $\beta_i'$, respectively.
By the above discussion, we have $\depth_{\beta_i'}(z_k) = 0$ or $\depth_{\beta_i'}(z_k) = \depth_{\beta_i}(w_k) < \depth_{\beta_i}(u)$.
In either case, we have $\depth_{\beta_i'}(z_k) \leq d - 1$ where $d = \max_i(\depth_{\beta_i}(u))$.
In other words, the depths of the newly introduced variables are smaller than the maximal depth of $u$.
Thus $N_j(\phi) = N_j(\phi')$ for $j > d$ and $N_d(\phi') < N_d(\phi)$ since $u$ is no longer instantiable.
Therefore $N(\phi') < N(\phi)$, as desired.

Next, suppose the variable $u$ returned by \Call{findInstantiation}{} was selected because $u$ occurs in $\beta_i^*$ and the sort $s$ of $u$ has only finitely many trees.
Then by the same arguments as before, new variables $\bar z$ are introduced in $\phi$ after adding $\exists \bar z \ldotp \gamma$ to $\alpha$.
However, since $\gamma$ describes a single value for $u$, every variable among $\bar z,u$ occurs on the left-hand side of an equation in the new $\alpha$.
Hence $u$ is no longer instantiable and none of the newly introduced variables $\bar z$ are.
Hence the number of instantiable variables decreases and thus $N(\phi') < N(\phi)$ for the new formula $\phi'$ as desired.

Next, suppose the variable $u$ returned by \Call{findInstantiation}{} was selected because $\fin(u)$ occurs in $\alpha$, $u$ occurs in $\beta_i^*$ and $s \in S_{FF}$.
Then the same argument as in the previous case can be applied.

Next, suppose the variable $u$ returned by \Call{findInstantiation}{} was selected because $s \in S_{FI}$ and there is a $\beta_j^*$ consisting only of $\fin()$-constraints, including $\fin(u)$.
After an instantiation of the form $\fin(u)$ and subsequent simplification, $\fin(u)$ will be removed from each $\beta_j^*$ since $\alpha$ is contained in $\beta_j$ by condition (2) of \cref{def-solved}, so $u$ is no longer instantiable.
Hence the number of instantiable variables, and thus $N(\phi)$, decreases.
After an instantiation of the form $\exists \bar z \ldotp \gamma$ describing an infinite value for $u$, the variable $u$ is also not instantiable anymore.
Since all the additional variables $\bar z$ occur on the left-hand side of an equation in $\exists \bar z \ldotp \gamma$, they are not instantiable either.
Hence the number of instantiable variables, and thus $N(\phi)$, decreases.

Altogether, $N(\phi)$ decreases after each instantiation step.
Therefore, only finitely many instantiations can happen.
\end{proof}

\begin{lemma}
\label{lem-conditions1-4}
In \cref{alg-solve-normal}, when \Call{removeUnreachableParts}{} is called, $\phi$ satisfies conditions (1) to (4) of a solved formula from \cref{def-solved}.
\end{lemma}
\begin{proof}
From the proof of correctness of the unmodified algorithm \cite[Property 4.6.3]{ddf08}, which works the same until the instantiation step, it follows that up until that point, $\phi$ satisfies conditions (1) to (3).
As soon as \Call{removeUnreachableParts}{} is called, (4) is satisfied because otherwise \Call{findInstantiation}{} would find a variable violating (4).
\end{proof}

\begin{lemma}
\label{lem-rule15-correct}
The function \Call{removeUnreachableParts}{} from \cref{alg-remove-unreachable} (Rule 15 in \cite[Section~4.6]{ddf08}) is still correct in the context of the extended algorithm.
\end{lemma}
\begin{proof}
As the previous lemma states, at the point where Rule 15 is applied, $\phi \equiv \lnot(\exists \bar x \ldotp \alpha \land \bigwedge_{i=1}^n \lnot(\exists \bar y_i \ldotp \beta_i))$ satisfies conditions (1) to (4) of \cref{def-solved}.
Repeating what was stated at the beginning of this section, when talking about reachability in a formula $\exists \bar x. \alpha$ where $\alpha$ is a basic formula, we mean reachability in $\alpha$ from the free variables of the whole formula.
As in the algorithm pseudocode, let
\begin{itemize}[noitemsep, nolistsep]
\item $\bar x'$ be the reachable variables of $\exists \bar x \ldotp \alpha$,
\item $\bar x'''$ the unreachable variables from $\bar x$ that occur on the LHS of an equation in $\alpha$,
\item $\bar x''$ the variables from $\bar x$ that are not in $\bar x'\bar x'''$,
\item $\alpha'$ be the reachable conjuncts of $\exists \bar x \ldotp \alpha$,
\item $\alpha''$ the unreachable $\fin()$-subformulae of $\exists \bar x \ldotp \alpha$,
\item $\alpha'''$ the unreachable equations of $\exists \bar x \ldotp \alpha$,
\item $\beta^*_i$ the result of removing $\alpha''$ from $\beta_i$,
\item $\bar y'_i$ the reachable variables among $\bar x'''\bar y_i$ in  $\exists \bar x'''\bar y_i \ldotp\beta^*_i$,
\item $\beta'_i$ the reachable conjuncts in $\exists \bar x'''\bar y_i \ldotp \beta^*_i$
\item $K \subseteq \{1,\dots,n\}$ the set of indices $i$ such that $i \in K$ if and only if no variable of $\bar x''$ occurs in $\beta'_i$.
\end{itemize}
Then the claim is that $\lnot(\exists \bar x \ldotp \alpha \land \bigwedge_{i=1}^n \lnot(\exists \bar y_i \ldotp \beta_i))$ is equivalent to $\lnot(\exists \bar x' \ldotp \alpha' \land \bigwedge_{i\in K} \lnot(\exists \bar y'_i \ldotp \beta_i'))$.

First note that $\lnot(\exists \bar x \ldotp \alpha \land \bigwedge_{i=1}^n \lnot(\exists \bar y_i \ldotp \beta_i))$ is equivalent to
\[ \lnot(\exists \bar x' \ldotp \alpha' \land (\exists \bar x'' \ldotp \alpha'' \land (\exists \bar x''' \ldotp \alpha''' \land \bigwedge_{i=1}^n \lnot(\exists \bar y_i \ldotp \beta_i)))) \]
because the variables $\bar x''$ can only occur in $\alpha''$ and the variables $\bar x'''$ can only occur in $\alpha'''$.
By the Unique Solution Axiom and since $\alpha'''$ is a solved formula, we have $\exists! \bar x''' \ldotp \alpha'''$ in the extended theory of trees.
According to \cite[Property 3.1.11]{ddf08}, the previous formula is equivalent to
\[ \lnot(\exists \bar x' \ldotp \alpha' \land (\exists \bar x'' \ldotp \alpha'' \land \bigwedge_{i=1}^n \lnot(\exists \bar x''' \ldotp \alpha''' \land \exists \bar y_i \ldotp \beta_i))). \]
By our variable convention, no variable names conflict, so the innermost existential can be pulled outside:
\[ \lnot(\exists \bar x' \ldotp \alpha' \land (\exists \bar x'' \ldotp \alpha'' \land \bigwedge_{i=1}^n \lnot(\exists \bar x'''\bar y_i \ldotp \alpha''' \land \beta_i))). \]
By condition (2) of \cref{def-solved}, the equations of $\alpha$ are included in each $\beta_i$.
In particular, $\alpha'''$ is part of each $\beta_i$, which simplifies the formula to
\[ \lnot(\exists \bar x' \ldotp \alpha' \land (\exists \bar x'' \ldotp \alpha'' \land \bigwedge_{i=1}^n \lnot(\exists \bar x'''\bar y_i \ldotp \beta_i))). \]
Note that $\beta^*_i \land \alpha'' \leftrightarrow \beta_i \land \alpha''$ by definition, so we can propagate $\alpha''$ into the innermost existential formulae: $\alpha'' \land \bigwedge_{i=1}^n \lnot(\exists \bar x''' \bar y_i \ldotp \alpha'' \land \beta_i) \leftrightarrow \alpha'' \land \bigwedge_{i=1}^n \lnot(\exists \bar x''' \bar y_i \ldotp \alpha'' \land \beta^*_i)$; and back out, yielding:
\[ \lnot(\exists \bar x' \ldotp \alpha' \land (\exists \bar x'' \ldotp \alpha'' \land \bigwedge_{i=1}^n \lnot(\exists \bar x'''\bar y_i \ldotp \beta^*_i))). \]
Since unreachable parts of a solved basic formula can be removed by the following \cref{lem-remove-unreachable}, this is equivalent to
\[ \lnot(\exists \bar x' \ldotp \alpha' \land (\exists \bar x'' \ldotp \alpha'' \land \bigwedge_{i=1}^n \lnot(\exists \bar y'_i \ldotp \beta'_i))). \]
Since a variable from $\bar x''$ can only occur in $\beta'_i$ if $i \notin K$, this is equivalent to
\[ \lnot \left(\exists \bar x' \ldotp \alpha' \land \left( \bigwedge_{i \in K} \lnot(\exists \bar y'_i \ldotp \beta'_i) \right) \land \left( \exists \bar x'' \ldotp \alpha'' \land \bigwedge_{i \notin K} \lnot(\exists \bar y'_i \ldotp \beta'_i) \right) \right). \]

To complete the proof, we show that the last conjunct $\exists \bar x'' \ldotp \alpha'' \land \bigwedge_{i \notin K} \lnot(\exists \bar y'_i \ldotp \beta'_i)$ is always true.
For this, it suffices to find valuations for $\bar x''$ satisfying $\alpha''$ but none of $\exists \bar y'_i \ldotp \beta'_i$ for $i \notin K$.
Note that since a variable from $\bar x''$ occurs in $\beta'_i$ for all $i \notin K$, each such $\beta'_i$ contains a conjunct of one of the following forms:
\begin{itemize}[noitemsep, nolistsep]
\item $\fin(v)$ for $v \in \bar x''$ and by the construction of $\beta^*_i$, $\fin(v)$ does not occur in $\alpha''$,
\item $v = \con f(\bar w)$ for $v \in \bar x''$,
\item $v = w$ where $v \in \bar x''$ and $v > w$, implying $w \notin \bar y'_i$,
\item $u = t$ where $v \in \bar x''$ occurs in $t$. Since it has to be reachable, that means that $\beta'_i$ contains the conjunction $\bigwedge_{j=1}^k w_j = t_j$ with $t_j$ containing $w_{j+1}$, $w_{k+1} \equiv v$, and $w_1 \notin \bar y'_i$.
Since the case $w_1 \in \bar x''$ was already handled in a previous case, we can assume without loss of generality that $w_1$ is a free variable.
\end{itemize}

The goal now is to find a valuation of $\bar x''$ that satisfies $\alpha''$ but that makes each of the above cases false, thus making $\lnot\exists y'_i \ldotp \beta'_i$ true.
Fix a valuation for the free variables of the formula.
Let $v: s$ be a variable from $\bar x''$.
\begin{itemize}[noitemsep, nolistsep]
\item If $\fin(v)$ occurs in $\alpha''$ then no $\beta'_i$ can contain $v = \con f(\bar w)$ because $v$ is not instantiable and thus $v$ would have to be properly reachable from itself, contradicting finiteness.
If $s$ only contains finitely many finite trees then $v$ occurs in no $\beta'_i$ because $v$ is not instantiable.
Then $v$ can be given any finite value to make $\alpha''$ true.
Otherwise, $v$ occurs only in equations of the form $u = t$ (reachable from some free variable $w_1$ as seen above) or $v = w$ in the $\beta_i'$.
In the former case, to make the the equation false, we pick a value for $v$ that is different from the one that is determined by the fixed value of $w_1$.
In the latter case, we pick a value $v$ different from the value of $w$.
Since $s$ contains infinitely many finite trees, it is possible to pick one as the value for $v$ that contradicts all those finitely many equations.
\item If $\fin(v)$ does not occur in $\alpha''$ and there is a $\beta'_j$ containing only $\fin()$-constraints, among them $\fin(v)$, then since $v$ is not instantiable, we have $s \notin S_{FI}$.
Thus there are infinitely many infinite trees of sort $s$.
Since there are only finitely many equations of the form $v = w$, or $u = t$ with $t$ containing $v$ (reachable from some free variable as above), or $v = t$ with $v$ properly reachable from itself in the $\beta'_i$, it is possible to find a value for $v$ that contradicts all of them by \cref{lem-contradict-recursive-constraints}.
\item If $\fin(v)$ does not occur in $\alpha''$ and there is no $\beta'_j$ containing only $\fin()$-constraints, among them $\fin(v)$, then there are two cases.
If $s \in S_{FF} \cap S_{FI}$ then since $v$ is not instantiable, no $\beta'_i$ contains $v$, and there are no constraints to contradict, or $\alpha$ contains an equation $v = t$, in which case each $\beta'_i$ also contains the same equation.
Hence all the $\beta'_i$ can be contradicted by picking a value different from $t$ for $v$.
Otherwise ($s \notin S_{FF} \cap S_{FI}$), there are infinitely many possible valuations for $v$ while there are only finitely many constraints of the form $u = t$ with $t$ containing $v$ (reachable from some free variable as above), or $v = w$, or $v = \con f(\bar w)$ with $v$ properly reachable from itself.
Again, \cref{lem-contradict-recursive-constraints} shows that it is possible to find a value for $v$ that contradicts all of these constraints.
\end{itemize}
We have shown above that by picking valuations for the variables from $\bar x''$ as described above, each $\beta'_i$ containing an equation is contradicted by the above valuation.
If a $\beta'_i$ contains only $\fin()$-constraints then at least one of those $\fin(v)$ is contradicted as described above.
This means that the above valuations for $v$ make all the $\beta'_i$ false, while satisfying $\alpha''$, independently of the values of the free variables.

This means that the formula
\[ \exists \bar x'' \ldotp \alpha'' \land \bigwedge_{i \notin K} \lnot(\exists \bar y'_i \ldotp \beta'_i) \]
is valid in the extended theory of trees.
\end{proof}

The above proof made use of the following lemma.

\begin{lemma}
\label{lem-remove-unreachable}
Let $\bar x$ be a vector of variables and $\alpha$ a solved basic formula.
Let $\bar x'$ be reachable variables in $\exists \bar x \ldotp \alpha$ and $\alpha'$ be the conjunction of equations and $\fin()$-formulae that are reachable in $\exists \bar x \ldotp \alpha$.
Then in the theory of trees, $\exists \bar x \ldotp \alpha$ is equivalent to $\exists \bar x' \ldotp \alpha'$.
\end{lemma}
\begin{proof}
Let $\bar x''$ be the unreachable variables in $\exists \bar x \ldotp \alpha$ that do not occur on the LHS of an equation of $\alpha$ and $\bar x'''$ be the unreachable variables which do.
Similarly, let $\alpha''$ be the conjunction of unreachable $\fin()$-formulae and $\alpha'''$ be the conjunction of unreachable equations in $\exists \bar x \ldotp \alpha$.
By the definition reachability, $\bar x''$ and $\bar x'''$ do not occur in $\alpha'$.
Hence $\exists \bar x \ldotp \alpha$ is equivalent to
\[ \exists \bar x'. \alpha' \land (\exists \bar x'' \ldotp \alpha'' \land (\exists \bar x''' \ldotp \alpha''')). \]
By the Unique Solution Axiom, $\exists! \bar x''' \ldotp \alpha'''$ holds for any valuation of the free variables in the extended theory of trees.
Hence the formula simplifies to
\[ \exists \bar x' \ldotp \alpha' \land (\exists \bar x'' \ldotp \alpha''). \]
Since $\alpha''$ contains only $\fin()$-formulae and since by condition (2) of \cref{def-basic}, they are all satisfiable, $\exists \bar x'' \ldotp \alpha''$ is true in the theory of trees as well.
Hence the original formula is equivalent to $\exists \bar x' \ldotp \alpha'$, as desired.
\end{proof}

\begin{table}
\centering
\begin{tabular}{r | r | r | r | r}
&\multicolumn{4}{c}{selector semantics}\\
Time to solve &\multicolumn{2}{c|}{standard} &\multicolumn{2}{c}{default values} \\
\hline
\hline
$<$ 1 ms &534 &13.35\% &197 &4.93\% \\
\hline
$<$ 10 ms &2241 &56.04\% &1415 &35.38\% \\
\hline
$<$ 100 ms &3247 &81.20\% &3224 &80.62\% \\
\hline
$<$ 1 s &3659 &91.50\% &3779 &94.50\% \\
\hline
$<$ 10 s &3816 &95.42\% &3929 &98.25\% \\
\hline
timed out ($>$ 10 s) &183 &4.58\% &70 &1.75\% \\
\hline
total & 3999 &100\% & 3999 &100\%
\end{tabular}
\caption{Results of the SMT-LIB QF\_DT benchmark suite: the number of benchmarks solved in the specified time limit (wall-clock time).
The measurements were made on a notebook computer with an Intel® Core™ i5-8250U CPU and 16 GB of RAM.}
\label{fig-benchmark-results}
\end{table}

\end{document}